\documentclass[graybox]{svmult}

\usepackage{mathptmx}       
\usepackage{helvet}         
\usepackage{courier}        
\usepackage{type1cm}        
\usepackage{makeidx}         
\usepackage{graphicx}        
\usepackage{multicol}        
\usepackage[bottom]{footmisc}

\usepackage{natbib}
\usepackage{amssymb,amsmath}
\usepackage{algorithm}
\usepackage{algorithmic}

\makeindex             

\begin{document}

\title*{Optimal Projections in the Distance-Based Statistical Methods}
\author{Chuanping Yu and Xiaoming Huo}
\institute{Chuanping Yu \at School of Industrial and Systems Engineering, Georgia Institute of Technology, Atlanta, GA \email{c.yu@gatech.edu}
\and Xiaoming Huo \at School of Industrial and Systems Engineering, Georgia Institute of Technology, Atlanta, GA \email{huo@gatech.edu}}
%
%
\maketitle

\abstract*{
xxx will copy and paste the contents below later.}

\abstract{This paper introduces a new way to calculate distance-based statistics, particularly when the data are multivariate.
The main idea is to pre-calculate the optimal projection directions given the variable dimension, and to project multidimensional variables onto these pre-specified projection directions;
by subsequently utilizing the fast algorithm that is developed in \citet{huo2016fast} for the univariate variables, the computational complexity can be improved from $O(m^2)$ to $O(n m \cdot \mbox{log}(m))$, where $n$ is the number of projection directions and $m$ is the sample size.
When $n \ll m/\log(m)$, computational savings can be achieved.
The key challenge is how to find the optimal pre-specified projection directions.
This can be obtained by minimizing the worse-case difference between the true distance and the approximated distance, which can be formulated as a nonconvex optimization problem in a general setting.
In this paper, we show that the exact solution of the nonconvex optimization problem can be derived in two special cases: the dimension of the data is equal to either $2$ or the number of projection directions.
In the generic settings, we propose an algorithm to find some approximate solutions.
Simulations confirm the advantage of our method, in comparison with the pure Monte Carlo approach, in which the directions are randomly selected rather than pre-calculated.}

\section{Introduction}
\label{sec:intro}
Distances are very important in statistics: a class of hypotheses testing methods are based on distances, such as the energy statistics  \citep{szekely2004testing}, the distance covariance \citep{szekely2007measuring, szekely2009brownian,lyons2013distance}, and many others.
This type of testing statistics usually belong to the class of U-statistics or the V-statistics \citep{mises1947asymptotic, hoeffding1992class, korolyuk2013theory}, which require the calculation of all pairwise distances within the sample.
When variables are univariate, assuming the sample size is $m$, both \citet{huo2016fast} and \citet{chaudhuri2018fast} proposed fast algorithms with computational complexity $O(m \mbox{log}(m))$ where $m$ is the sample size.
Recall that the computational complexity is $O(m^2)$ when the statistics are computed directly based on their definitions.
When variables are multivariate, especially when they are high-dimensional, the calculation of the pairwise distances among these multivariate variables can not be implemented directly by the algorithm in \cite{huo2016fast}, and therefore becomes a potential bottleneck.
Our paper is aimed at reducing the computation complexity in the multivariate case by projecting the variables along a set of pre-specified optimal directions.
When the number of pre-specified optimal directions $n \ll m/\log(m)$, computational savings can be achieved, since the computational complexity is $O(nm\cdot\mbox{log}(m))$, which would be less than $O(m^2)$.

We use the energy distances \citep{szekely2004testing} as an example to solidify our motivation.
The energy statistic is used to test the equality between two distributions.
More precisely, suppose $X_1, . . . , X_{n_1} \in\mathbb{R}^p, p \geq 1$ are independent and identically distributed (i.i.d.), sampled from the distribution $F_X$, and $Y_1, . . . , Y_{n_2} \in\mathbb{R}^p$  are i.i.d., sampled from the distribution $F_Y$.
The two-sample test statistic (also called the energy statistic) for testing the two-sample hypothesis
$$
H_0: F_X=F_Y
$$
is defined as \citep{szekely2004testing}:
\begin{equation}
\label{energy}
\mathcal{E}_{n_1,n_2} \triangleq
\frac{2}{n_1 n_2}\sum\limits_{i=1}^{n_1}\sum\limits_{j=1}^{n_2}\left\|X_i-Y_j\right\|
-\frac{1}{n_1^2}\sum\limits_{i=1}^{n_1}\sum\limits_{k=1}^{n_1}\left\|X_i-X_k\right\|
-\frac{1}{n_2^2}\sum\limits_{j=1}^{n_2}\sum\limits_{k=1}^{n_2}\left\|Y_j-Y_k\right\|,
\end{equation}
where $\left\|X_i-Y_j\right\|, \left\|X_i-X_k\right\|, \left\|Y_j-Y_k\right\|$ are the distances from  the two samples.
Note that the statistic $\mathcal{E}_{n_1,n_2}$ solely depends on three types of inter-point distances:
$\left\|X_i-Y_j\right\|,
\left\|X_i-X_k\right\|,
\left\|Y_j-Y_\ell\right\|,
i,k = 1,\ldots,n_1,
j,\ell=1,\ldots,n_2.$
Denote $m = n_1+n_2$.
\cite{huang2017efficient} have showed that it can be efficiently computed with computational complexity $O(m\mbox{log}(m))$ in the univariate case (i.e., $p=1$).

When $X_i$\rq s and $Y_j$\rq s are multivariate (i.e., we have $p>1$), random projections have been proposed to find a fast approximation to the statistic $\mathcal{E}_{n_1,n_2}$.
For example, \cite{huang2017efficient} gave a fast algorithm that is based on random projections, which can achieve $O(n m \cdot \mbox{log}(m))$ computational complexity, where $n$ is the number of random projections.
Note that the approach in \cite{huang2017efficient} is a pure Monte Carlo approach.
The recent advances in the quasi-Monte Carlo methods \citep{niederreiter1992random, morokoff1995quasi} have demonstrated that in some settings, utilizing pre-determined projections can lead to better performance than the completely random ones in the pure Monte Carlo approach.
Quasi-Monte Carlo methods sometimes enjoy faster rate of convergence, e.g., \cite{asmussen2007stochastic}.


Our approach turns a distance calculation in a multivariate situation to the one in a univariate situation.
The proposed approach
\begin{enumerate}\setlength\itemindent{1em}
  \item[{\bf P1.}] first projects each multivariate variable along some pre-specified optimal directions to corresponding one-dimensional subspaces (the projected values are univariate),
  \item[{\bf P2.}] then the sum of the $\ell_1$ norm of the projected values is used to approximate the associated distance in the multivariate setting.
\end{enumerate}
More specifically, let's suppose the multivariate variable is $v = (v_1,...,v_p) \in \mathbb{R}^p$.
Recall that the norm of $v$ is
$$
||v||=\sqrt{\sum\limits_{i=1}^p v_i^2}.
$$
For $n \geq 1$, our objective is to identify the projection directions, which can be represented by vectors $u_1, u_2, ..., u_n \in \mathbb{R}^p$, and a predetermined constant $C_n \in \mathbb{R}$, such that for any $v\in  \mathbb{R}^p$, we have
\begin{equation}
||v||\approx C_n\sum\limits_{i=1}^n\left|u_i^Tv\right|.
\end{equation}
Consequently in step {\bf P2.}, when one needs to compute a distance $\|X_i - Y_j\|$, one can alternatively compute $C_n \sum\limits_{i=1}^n \left|u_i^T X_i - u_i^T Y_j\right|$.
Note that $u_i^T X_i$ and $u_i^T Y_j$ are univariate.
Therefore, the fast algorithm in the one-dimensional case can be utilized.

We continue with the example of the energy distances.
Recall that the pre-specified directions are supposed to be $u_1, . . . , u_n$.
The projected values of the corresponding multivariate variables then become
\begin{eqnarray*}
X_{wi} &=& u_w^TX_i\in\mathbb{R}, w=1, . . ., n; i = 1, . . ., n_1; \mbox{ and } \\
Y_{wj} &=& u_w^TY_j\in\mathbb{R}, w=1, . . ., n; j = 1, . . ., n_2.
\end{eqnarray*}
The distance between any two multivariate variables can be approximated by the sum of these projections multiplying by a constant:
$$
\|X_i - Y_j\| \approx C_n\sum\limits_{w=1}^n|X_{wi}-Y_{wj}|.
$$
Therefore, the statistic $\mathcal{E}_{n_1,n_2}$ in \eqref{energy} can be approximated by
\begin{eqnarray}
\mathcal{E}_{n_1,n_2} &\approx&
C_n\Big(\frac{2}{n_1n_2}\sum\limits_{i=1}^{n_1}\sum\limits_{j=1}^{n_2}\sum\limits_{w=1}^{n}\left\|X_{wi}-Y_{wj}\right\|
-\frac{1}{n_1^2}\sum\limits_{i=1}^{n_1}\sum\limits_{k=1}^{n_1}\sum\limits_{w=1}^{n}\left\|X_{wi}-X_{wk}\right\| \nonumber \\
&& -\frac{1}{n_2^2}\sum\limits_{j=1}^{n_2}\sum\limits_{k=1}^{n_2}\sum\limits_{w=1}^{n}\left\|Y_{wj}-Y_{wk}\right\|\Big) \nonumber \\
&=& C_n\Big(\frac{2}{n_1n_2}\sum\limits_{i=1}^{n_1}\sum\limits_{j=1}^{n_2}\sum\limits_{w=1}^{n}\left|X_{wi}-Y_{wj}\right|
-\frac{1}{n_1^2}\sum\limits_{i=1}^{n_1}\sum\limits_{k=1}^{n_1}\sum\limits_{w=1}^{n}\left|X_{wi}-X_{wk}\right| \label{eq:energy_p} \\
&&-\frac{1}{n_2^2}\sum\limits_{j=1}^{n_2}\sum\limits_{k=1}^{n_2}\sum\limits_{w=1}^{n}\left|Y_{wj}-Y_{wk}\right|\Big). \nonumber
\end{eqnarray}
The second equation is true because in the one-dimensional case, the $\ell_2$ norm becomes the absolute value.
Then one can apply the fast algorithms for univariate variables to calculate the energy statistic in \eqref{eq:energy_p}.

\textbf{Remark:} \textit{Our method is not restricted to the calculation of the energy statistic, or other distance-based statistics.
It can also be applied to the calculation of the distance-based smooth kernel functions.}

In this paper, we first give a detailed description of our strategy to find the optimal pre-specified projection directions.
We formulate the searching for optimal projection directions problem as a minimax optimization problem.
Let $\left\{u_1, u_2, \cdots, u_n\right\}$ denote the optimal set of projection directions, they should minimize the worst-case difference between the true distance and the approximate distance.
Equation \eqref{eq:intuit} below shows this idea in the mathematical form:
\begin{equation}
\label{eq:intuit}
\min\limits_{\begin{subarray}{c}
C_n, u_i: \\
\|u_i\|=1, i=1,\cdots, n\end{subarray}}
\max\limits_{v: \|v\|_2\le 1}
\left|C_n\sum_{w=1}^n\left|u_w^T v \right|-\|v\|\right|.
\end{equation}
Discussion on how to solve the above problem is presented in Section \ref{sec:formulation}.

In general, the problem in \eqref{eq:intuit} is a nonconvex optimization problem, which is potentially NP-hard.
We found that in two special cases, the optimal directions can be derived analytically:
(a) the $2$-dimensional case and
(b) when the dimension is equal to the number of projections.
More details on these two special cases are presented in Section \ref{sec:anal}.
In general cases, we propose a greedy algorithm to find the projection directions.
Note that the greedy algorithm terminates at a local optimal solution to \eqref{eq:intuit}.
In this case, we cannot theoretically guaranteed that the found directions correspond to the global solution to the problem in \eqref{eq:intuit}, which is the case in most nonconvex optimization problems.
At the same time, the simulations show that our approach can still outperform the pure Monte Carlo approach in many occasions.

The rest of this paper is organized as follows.
Section \ref{sec:formulation} shows the formulation of our problem.
Section \ref{sec:anal} provides the analytical solutions to the problem in \eqref{eq:intuit} in two special cases.
Section \ref{sec:general} presents the numerical algorithm for the general cases.
In Section \ref{sec:simulation}, the simulation results of our method are furnished.
Section \ref{sec:conc} contains the conclusion and a summary of our work.
All the technical proofs are relegated to the appendix (Section \ref{sec:appendix}).

We adopt the following notations.
Throughout this paper, we use $p$ to denote the dimension of the data.
The sample size is denoted by $m$.
The number of projections is denoted by $n$.

\section{Problem formulation}
\label{sec:formulation}

As mentioned above, in order to estimate the distance between two multivariate variables, we project them onto some pre-specified one-dimensional linear subspaces.
We present details in the following.
Suppose the multivariate variable is $v = (v_1,...,v_p) \in \mathbb{R}^p$.
Recall that the norm of vector $v$ is
$$
||v||=\sqrt{\sum\limits_{i=1}^pv_i^2}.
$$
Our objective is to design $u_1, u_2, ..., u_n \in \mathbb{R}^p$, for $n \geq 1,$ and  $C_n \in \mathbb{R}$, such that for any $v\in  \mathbb{R}^p$, we have
\begin{equation}\label{appr}
||v||\approx C_n\sum\limits_{i=1}^n\left|u_i^Tv\right|.
\end{equation}
We would like to turn a distance (i.e., norm) of a multivariate variable $v$ into a weighted sum of the absolute values of some of its one dimensional projections (i.e., $u_i^Tv$\rq s), knowing that the one dimensional projections may facilitate efficient numerical algorithms.

Without loss of generality, we may assume $||v|| = 1$.
The approximation problem in \eqref{appr} can be formulated into the following problem:
\begin{equation}\label{obj}
\min\limits_{C_n, u_1, ..., u_n}\max\limits_{v: ||v||_2=1}\left|C_n\sum_{i=1}^n\left|u_i^Tv\right|-1\right|.
\end{equation}
In words, we would  like to select $u_1,..., u_n$ and $C_n$ such that the approximation in \eqref{appr} has the minimal discrepancy in the worst case.
One can verify that the problem in \eqref{obj} and the problem in \eqref{eq:intuit} share the same solution.

To solve the problem in \eqref{obj}, the following two quantities are needed.
For fixed $u_1, u_2, ..., u_n$, we define
\begin{eqnarray}
\label{eq:vmax}
V_{\max} &=& \max\limits_{v:||v||_2=1}\sum\limits_{i=1}^n\left|u_i^Tv\right|, \\
V_{\min} &=& \min\limits_{v:||v||_2=1}\sum\limits_{i=1}^n\left|u_i^Tv\right|, \label{eq:vmin}
\end{eqnarray}
where $V_{\max}$ and $V_{\min}$ are the maximum and minimum of
$\sum\limits_{i=1}^n\left|u_i^Tv\right|$
among all possible $v$ under the constraint $||v||_2=1$, respectively.
With these two quantities (i.e., $V_{\max}\mbox{ and }V_{\min}$), we have the following result.
\begin{theorem}
\label{Cn}
For given $u_1, u_2,...,u_n\in\mathbb{R}^p$, the optimal value for $C_n$ in the problem \eqref{obj} is
$$
C_n = \frac{2}{V_{\min}+V_{\max}}.
$$
Furthermore, the solutions of $u_1, u_2,...,u_n$ in problem \eqref{obj} are identical to the solutions to the following problem:
\begin{equation}
\label{obj_f}
\max\limits_{\begin{subarray}{c}u_1,...,u_n:\\ ||u_i||=1,\forall i, 1\leq i\leq n\end{subarray}}\frac{V_{\min}}{V_{\max}}.
\end{equation}
\end{theorem}
The above theorem indicates that the minimax problem in \eqref{obj} is equivalent to the maximization problem in \eqref{obj_f}.
Note that in general, both problems are nonconvex, therefore potentially NP-hard.
In our analysis, we found that both formulations (in \eqref{obj} and \eqref{obj_f}) are convenient in various steps of derivation.
Both of them are used in later analysis.

\section{Derivable analytical results}
\label{sec:anal}

We present the two special cases where analytical solutions are derivable.
When the dimension is $2$ (i.e., $p=2$), we show in Section \ref{sec:p=2} that an analytical solution to the problem in \eqref{obj_f} is available.
In Section \ref{sec:p=n}, we present another case (when the dimension of the data is equal to the number of projections, that is we have $n=p$) where an analytic solution to the problem in \eqref{obj_f} is derivable.

\subsection{Special case when the dimension is $2$}
\label{sec:p=2}

When the multivariate variables are two-dimensional, we can get the exact optimal projections that minimize the worse-case discrepancy.
The following theorem describes such a result.
\begin{theorem}
\label{dim2}
When $p = 2$, the $2$-dimensional vectors $u_1, u_2, ..., u_n$ can be represented by
$$
u_i = e^{\sqrt{-1}\theta_i}, i = 1,...,n.
$$
The optimal solution in \eqref{obj_f} has the form
\begin{equation}
\label{optdim2}
\theta_i=\frac{(i-1)\pi}{n} + k_i\pi, i = 1,...,n
\end{equation}
where each $k_i\in\mathbb{N}$.
\end{theorem}
Specially, when $n$ is odd, the optimal solutions can be represented by the equally spaced points on the circle.
Furthermore, we can get the error rate in the $2$-dimensional case, as in the following theorem.
\begin{theorem}
\label{ese_O}
If $u_1, \cdots, u_n$ are chosen according to Theorem \ref{dim2}, we have
$$
\mathop{\mathbb{E}}\limits_{v\sim \mbox{Unif}(S^1)}\left\{\left|C_n\sum\limits_{i=1}^n\left|u_i^Tv\right|-1\right|^2\right\} = O\left(\frac{1}{n^2}\right).
$$
\end{theorem}

\textbf{Remark:}\textit{Theorem \ref{ese_O} can be used as a guidance of choosing the number of directions.
Assume we would like to control the squared error to be $\epsilon$.
Then, we can get $\frac{1}{n^2} = \epsilon$, and therefore the number of directions should be larger than $\frac{1}{\sqrt{\epsilon}}$.}

In the above theorem, the random vector $v$ is sampled independently from the Uniform distribution on the unit circle $S^1$.
Note that the squared error rate is $O(1/n^2)$.
The following theorem presents the corresponding rate for the pure random projections.
\begin{theorem}
\label{ese_M}
If $u_1, \cdots, u_n$ are selected base on Monte Carlo, we have
$$
\mathop{\mathbb{E}}\limits_{u_i, v\sim \mbox{Unif}(S^1)}\left\{\left|C_n\sum\limits_{i=1}^n\left|u_i^Tv\right|-1\right|^2\right\} = O\left(\frac{1}{n}\right).
$$
\end{theorem}
In the above theorem, both random vector $v$ and vectors $u_i$'s are independently sampled from the Uniform distribution on the unit circle ($S^1$).
The squared error rate in the pure Monte Carlo case is $O(1/n)$.
These two theorems illustrate the theoretical advantage of adopting the pre-calculated projection directions (in relative to the random projections).
Such a phenomenon has been discovered in the literature regarding the quasi-Monte Carlo methodology.

\subsection{Second special case with provable result}
\label{sec:p=n}

When the dimension is larger than $2$, the problem in \eqref{obj} is challenging.
There is some potentially relevant literature in mathematics, such as the searching for algorithms to locate the equally-distributed points on the surfaces of some high-dimensional spheres \citep{sloan2004extremal, hesse2010numerical, brauchart2014qmc}.
We fail to locate the exact solutions to our problem.

Our analysis indicates that when the number of projections is equal to the dimension, an analytical solution to the problem in \eqref{obj} is derivable.
We present details in the following.
To derive our analytical solution in a special case, we need to revisit two quantities, $V_{\min}$ and $V_{\max}$, which have been introduced in \eqref{eq:vmax} and \eqref{eq:vmin}.
The following lemma is about $V_{\max}$.
\begin{lemma}\label{V_max}
For fixed $u_1, u_2,...,u_n\in\mathbb{R}^p$, we have
\begin{equation}\label{Vmax}
V_{\max} = \max\limits_{s_i\in\{1, -1\}}\left\|\sum\limits_{i=1}^ns_iu_i\right\|.
\end{equation}
\end{lemma}
Lemma \ref{V_max} points out a way to calculate $V_{\max}$, that is, given binary $s_i$'s, finding out the linear combination $\sum\limits_{i=1}^ns_iu_i$ with the maximal norm out of the all possible $2^n$ linear combinations.
Let $\{s_i^{\max}\in\{1,-1\}: i = 1, ..., n\}$ denote the solution for \eqref{Vmax} when $u_1,\cdots, u_n$ are given.
The Algorithm \ref{alg:s} formally presents the aforementioned approach.
Assume we are in the $k$-th loop, where the $u_j$\rq s are known, which are denoted by $u_1^{(k)}, u_2^{(k)}, ..., u_n^{(k)}$.
Let $s_i^{(k)}$\rq s denote the $s_i$'s that can achieve $V_{\max}$ in the $k$-th loop.
We have the Algorithm \ref{alg:s}.
\begin{algorithm}[htbp]
\caption{Find $s_i^{\max}$\rq s in the $k$-loop}
\label{alg:s}
\begin{algorithmic}[1]
\REQUIRE{Unit vectors $u_1^{(k)}, u_2^{(k)}, \dots, u_n^{(k)}\in S^{p-1}$ are given.}
\ENSURE {$s_i^{(k)}$\rq s.}
\FORALL {binary combination of $s_i^{(k)}$\rq s}
\STATE Calculate the value $\left\|\sum\limits_{i=1}^ns_iu_i\right\|$.
 \ENDFOR
\STATE The binary combination that can make the value of $\left\|\sum\limits_{i=1}^ns_iu_i\right\|$ be the maximum among all the possible values,
is the $s_i^{\max}$\rq s, which is denoted as $s_i^{(k)}$\rq s.
\end{algorithmic}
\end{algorithm}

As for $V_{\min}$, suppose $v_{\min}$ is a minimizer of $V_{\min}$.
We have the following property for $v_{\min}$.
\begin{lemma}
\label{prop}
For fixed $u_1, u_2,...,u_n\in\mathbb{R}^p$, if $\Omega$ is an intersection of $S^{p-1}$ and a linear subspace with at least 2 dimensions, then the solution to the minimization problem
$$
\min_{v\in \Omega} f(v) = \sum\limits_{i=1}^n\left|u_i^Tv\right|
$$
must have $u_j^Tv_{\min}=0$ for at least one $j$ $(1\leq j\leq n)$.
\end{lemma}
Geometrically, the above lemma indicates that vector $v_{\min}$ should be orthogonal to at least one of the projection vector $u_j$.
For vector $v_{\min}$, we will need the following definition to further our derivation.
\begin{definition}[maximal subset]
\label{def:max-subset}
We call $\Omega(v_{\min})$ a maximal subset of the set $\{u_1,...,u_n\}$ if it satisfies
$$
\Omega(v_{\min}) = \left\{u_j: u_j^Tv_{\min}=0\right\} \subset \{u_1,...,u_n\},
$$
and it cannot be a strict subset for another $\Omega(v^\prime_{\min})$ where $v^\prime_{\min}$ is a minimizer that is different from $v_{\min}$.
\end{definition}
Lemma \ref{prop} ensures that the set $\Omega(v_{\min})$ cannot be empty.
The following lemma shows that the linear subspace that is spanned by the elements of $\Omega(v_{\min})$ must have certain dimensions.
\begin{lemma}
\label{Omega_0}
If $\Omega(v_{\min})$ is a maximal subset of $u_1,...,u_n$,
we must have
$$
\mbox{rank}\left(\Omega(v_{\min})\right)=p-1,
$$
for any minimizer $v_{\min}$.
\end{lemma}
Recall $p$ is the dimension of the data.
The above lemma essentially states that the space that is spanned by the elements of $\Omega(v_{\min})$ is the orthogonal complement subspace of the one-dimensional space that is spanned by the vector $v_{\min}$.

One direct corollary of Lemma \ref{Omega_0} is that the cardinality of the set $\Omega(v_{\min})$ is at least $p-1$.
Consequently, the total number of possible sets (of $\Omega(v_{\min})$) is no more than $n \choose p-1$.
This inspires us to use Algorithm \ref{alg:vmin} to find $v_{\min}$ as well as $\Omega(v_{\min})$ if all the $u_j$\rq s are given.
Here suppose we are in the $k$-th loop where the $u_j$\rq s are known, which are $u_1^{(k)}, u_2^{(k)}, ..., u_n^{(k)}$.

\begin{algorithm}[htbp]
\caption{Find $v_{\min}$ and $\Omega(v_{\min})$ in the $k$-loop}
\label{alg:vmin}
\begin{algorithmic}[1]
\REQUIRE{Unit vectors $u_1^{(k)}, u_2^{(k)}, \dots, u_n^{(k)}\in S^{p-1}$ are given.}
\ENSURE {$v^{(k)}$ and $\Omega(v^{(k)})$.}
\FORALL {$(p-1)$ combination of $u_i^{(k)}$\rq s, denoted as $S^u_t$}
\WHILE {$\mbox{rank}(S^u_t)<p-1$}
\STATE Add another $u_j$ that is not in the set $S_t^u$;
\ENDWHILE
\STATE Find the orthogonal direction of the set $S_t^u$, which is one of the candidates of $v^{(k)}$, denoted as $v^{(k)}_t$,
 and calculate the value of $f(v^{(k)}_t) = \sum\limits_{i=1}^{n} \left|\left(u_i^{(k)}\right)^Tv^{(k)}_t \right|$.
 \ENDFOR
\STATE The $v^{(k)}_t$, that can make the value of $f(v^{(k)}_t)$ be the minimum among all the possible $f(v^{(k)}_t)$ values,
is the $v_{\min}$, which is denoted as $v^{(k)}$,
and the corresponding $S_t^u$ set is the set $\Omega(v_{\min})$, which is denoted as $\Omega(v^{(k)})$.
\end{algorithmic}
\end{algorithm}

From Lemma \ref{Omega_0} we can get the exact solution for the special case when the number of projection directions is equal to the dimension of the multivariate variables, which is described in the following theorem.
\begin{theorem}
\label{dim=n}
When the number of projections is equal to the dimension of the data, i.e., we have $n = p$, the optimal solution in \eqref{obj_f} satisfies the following condition:
\begin{equation}\label{optdim=n}
u_i^Tu_j=0, \forall i \not= j.
\end{equation}
The above is equivalent to stating that the set $\left\{u_1, u_2, \cdots, u_n\right\}$ forms an orthonormal basis in $\mathbb{R}^p$.
\end{theorem}


\section{Numerical approach in general cases}
\label{sec:general}

When $p>2$ and $n\neq p$, we propose an algorithm to identify the optimal projections $u_1, u_2, ..., u_n$, such that they solve \eqref{obj_f}.
Per Lemma \ref{V_max} and the definition of $s_i^{\max}$\rq s,
the $V_{\max}$ can be written as:
$$
V_{\max} = \left\|\sum\limits_{i=1}^ns_i^{\max}u_i\right\|.
$$
According to Lemma \ref{Omega_0}, we have
\begin{eqnarray*}
V_{\min} = \sum\limits_{i=1}^n\left|u_i^Tv_{\min}\right|
&=& \sum\limits_{u_i\in\Omega(v_{\min})}\left|u_i^Tv_{\min}\right|
+ \sum\limits_{u_i\not\in\Omega(v_{\min})}\left|u_i^Tv_{\min}\right| \\
&=& \sum\limits_{u_i\not\in\Omega(v_{\min})}\left|u_i^Tv_{\min}\right|.
\end{eqnarray*}
So when $u_1,\cdots, u_n$ are given, $\frac{V_{\min}}{V_{\max}}$ can be written as
\begin{equation}\label{VminVmax}
\frac{V_{\min}}{V_{\max}}
= \frac{\sum\limits_{u_i\not\in\Omega(v_{\min})}\left|u_i^Tv_{\min}\right|}{\left\|\sum\limits_{i=1}^ns_i^{\max}u_i\right\|},
\end{equation}
where $v_{\min}$ and $\Omega(v_{\min})$ are defined in Section \ref{sec:p=n}.
We assume that the set $\Omega(v_{\min})$ corresponds to the minimum over all $n \choose p-1$ possible sets, and ($s_i^{\max}$)\rq s maximize the norm of $\sum\limits_{i=1}^n s_i^{\max}u_i$.

We use a method that is similar to the coordinate descent algorithm \citep{nesterov2012efficiency,wright2015coordinate} to search for the optimal solutions of \eqref{obj_f}.
Details of our algorithm can be found in Algorithm \ref{alg:general}.
The optimal solution can be achieved in circular iterations: maximizing \eqref{VminVmax} with respect to one $u_i$, while the others are fixed.
We then iteratively maximize the objective function in \eqref{VminVmax} until the value of the objective function \eqref{VminVmax} cannot be increased.

We derive the iteration strategy in the following.
Let $v^{(k)}$ be the minimizer of $\sum\limits_{i=1}^n\left|u_i^Tv\right|$ at the $k$th iteration.
Let $\Omega^{(k)}$ denote the minimum over all $n \choose p-1$ possible sets at the $k$th iteration.
For any $u_j^{(k)}\not\in\Omega^{(k)}$, without loss of generality, we assume that $u_1\not\in\Omega^{(k)}$.
The objective function in \eqref{VminVmax} can be written as
\begin{equation}\label{VminVmax_iter}
\frac{V_{\min}}{V_{\max}}
= \frac{\left|u_1^Tv^{(k)}\right|
+\sum\limits_{i>1, u_i\not\in\Omega^{(k)}}\left|u_i^Tv^{(k)}\right|}{\left\|s_1^{\max}u_1
+ \sum\limits_{i=2}^ns_i^{\max}u_i\right\|}.
\end{equation}
Without loss of generality, we can assume $s_1^{\max} = 1$.
This is because, recalling that ($s_i^{\max}$)'s are binary, we have
$$
\left\|s_1^{\max}u_1 + \sum\limits_{i=2}^ns_i^{\max}u_i\right\|
= \left\|u_1 + \sum\limits_{i=2}^ns_1^{\max}s_i^{\max}u_i\right\|.
$$
The expression in \eqref{VminVmax_iter} can then be rewritten as
\begin{equation}
\label{eq:u1}
\frac{\left|u_1^Tv^{(k)}\right|+A}{\left\|u_1 +B\right\|},
\end{equation}
where
$$
A = \sum\limits_{i>1, u_i\not\in\Omega^{(k)}}\left|u_i^Tv^{(k)}\right|, \mbox{ and }
B = \sum\limits_{i=2}^ns_i^{\max}u_i.
$$
Note that quantities $A$ and $B$ do not depend on $u_1$.
Our objective is to derive a strategy to maximize the quantity in \eqref{eq:u1} as a function of the vector variable $u_1$.

We first solve a constrained version of the above maximization problem.
We define
$\Sigma(v, \theta) = \left\{x:\|x\| = 1, \left<x,v\right>=\theta\right\}$, for any fixed $\theta\in[0,\pi)$,
where $\left<\cdot,\cdot \right>$ denote the angle between two vectors.
Conditioning on $u_1 \in \Sigma(v, \theta)$, and $v=v^{(k)}$, maximizing the function in \eqref{eq:u1} is equivalent to maximizing the following function:
\begin{equation}\label{opt_u1}
\frac{\left|\cos\theta\right|+A}{\left\|u_1 +B\right\|}.
\end{equation}
Note that the numerator is not a function of $u_1$.
Consequently, it is equivalent to minimizing
$$
\left\|x+B\right\|, \mbox{ where } x\in\Sigma(v, \theta).
$$
The following lemma presents an analytical solution to the above minimization problem.
\begin{lemma}
\label{u_1}
Given a vector $B$, a constant $\theta \in [0, \pi)$, and a unit-norm vector $v$, the solution to the following problem
\begin{equation}\label{u1}
\min_{x:\|x\|=1,\left<x,v\right>=\theta}\left\|x+B\right\|^2
\end{equation}
is
\begin{equation}\label{x}
x = v\cos\theta + \frac{|\sin\theta|}{\sqrt{B^TB-(v^TB)^2}}\left[(v^TB)v-B\right].
\end{equation}
\end{lemma}
Using the solution in \eqref{x} to substitute the  $u_1$ in \eqref{opt_u1}, we have
\begin{equation}\label{opt_u1_2}
\frac{\left|\cos\theta\right|+A}{\left\|u_1 +B\right\|}
= \frac{\left|\cos\theta\right|+A}{\left\|v\cos\theta +B
+ \frac{|\sin\theta|}{\sqrt{B^TB-(v^TB)^2}}\left[(v^TB)v-B\right]\right\|}.
\end{equation}
Maximizing \eqref{opt_u1} with respect to $\theta$ is equivalent to maximizing \eqref{opt_u1_2}.
For fixed $A$, $B$, and $v$, the right hand side of \eqref{opt_u1_2} is a function of $\theta$.
The following Theorem \ref{gtheta} gives the solution to the above problem.
\begin{theorem}
\label{gtheta}
The solutions of maximizing \eqref{opt_u1} with respect to $\theta$ are the zeros of the following function:
\begin{equation}\label{eq:g}
g(\theta) =
\left\{ \begin{array}{ll}
\sqrt{B^TB}\left[\cos\alpha
+ A\cos(\alpha-\theta) - \sin\theta\sin(\alpha - \theta)\right] & \\
-(1+B^TB)\sin\theta, & \mbox{ if } \theta\in[0,\frac{\pi}{2}),\\
\sqrt{B^TB} \left[-\cos\alpha + A\cos(\alpha-\theta) + \sin\theta\sin(\alpha - \theta)\right] & \\
+ (1+B^TB)\sin\theta  &\mbox{ if } \theta\in[\frac{\pi}{2},\pi),
\end{array}
\right.
\end{equation}
where $\alpha$ satisfies
$\sin\alpha = \frac{v^TB}{\sqrt{B^TB}},$ and
$\cos\alpha = \frac{\sqrt{B^TB-(v^TB)^2}}{\sqrt{B^TB}}.$
\end{theorem}
The above theorem indicates that one can adopt a line search algorithm to compute for $\theta$.

Based on all the above, the Algorithm \ref{alg:general} (below) furnishes a coordinate ascent scheme to maximize the objective in \eqref{obj_f}.
\begin{algorithm}[htbp]
\caption{Optimal projection algorithm}
\label{alg:general}
\begin{algorithmic}[1]
\REQUIRE {Set a threshold $\Delta>0$, initial unit vectors $u_1^{(0)}, u_2^{(0)}, ..., u_n^{(0)}\in S^{p-1}$.
Thus, by Algorithm \ref{alg:s} and \ref{alg:vmin}, we can get the corresponding values $v^{(0)}$,$\Omega^{(0)}(v^{(0)})$, and $s_i^{(0)}$\rq s.}
\REPEAT
\STATE In the $k$-th loop, suppose the previous $u_1^{(k-1)}, u_2^{(k-1)}, \ldots, u_n^{(k-1)}$ are known.
\FORALL{$u_j^{(k-1)}\not\in\Omega^{(k-1)}(v^{(k-1)})$}

\STATE Find the zeros of the function $g(\theta)$ in \eqref{eq:g} in Theorem \ref{gtheta}, where
$v = v^{(k-1)},$
$B = \sum\limits_{i\not=j}s_j^{(k-1)} s_i^{(k-1)} u_i^{(k-1)}$,
and denote the zeros as $\theta^*$.
\STATE According to Lemma \ref{u_1}, the new $u_j^{(k)}$ would be
$v\cos\theta^* + \frac{|\sin\theta^*|}{\sqrt{B^TB-(v^TB)^2}}\left[(v^TB)v-B\right].$
\STATE By Algorithm \ref{alg:s} and \ref{alg:vmin}, we can get the corresponding values $v^{(k)}$,$\Omega^{(k)}(v^{(k)})$, and $s_i^{(k)}$\rq s,
based on the newly updated $u_j$\rq s, which also give us the value of $V_{\min}$ and $V_{\max}$.
\STATE Compute $V_{\min}/V_{\max}$.
\ENDFOR
\STATE Pick the $u_j^{(k)}\not\in\Omega^{(k-1)}(v^{(k-1)})$ that gives the maximal value of $V_{\min}/V_{\max}$ in the above loop.

\IF {The value of $V_{\min}/V_{\max}$ decreases}
\STATE Go back to $u_j^{(k-1)}$.
\ENDIF
\UNTIL{The increment of $V_{\min}/V_{\max}$ is less than $\Delta$}.
\end{algorithmic}
\end{algorithm}


\section{Simulations}
\label{sec:simulation}

In the previous section, the optimal projections for both the special cases and the general case are provided.
The simulations will follow the same order.
The simulations are about the comparison of the Monte Carlo method and our method for the special cases and then for a general case.

According to \cite{huang2017efficient}, Monte Carlo method is to  select some random directions, denoted as $w_i$, $i=1, \dots, n$, on the unit sphere $S^{p-1}$ and project the vector we would like to estimate, that is $v$, along these directions, so the norm of the vector $v$ could be estimated as
$$
\|v\| \approx C^\prime_p\frac{1}{n}\sum\limits_{i=1}^n |w_i^Tv|,$$
where $C^\prime_p = \frac{\sqrt{\pi}\Gamma(\frac{p+1}{2})}{\Gamma(\frac{p}{2})}$.

In all the experiments, we randomly select $100$ unit vectors on the sphere as the vectors that we would like to estimate, in order to get the mean squared error for comparison between the Monte Carlo method and the method we propose.

\subsection{When the dimension is $2$}

When the dimension is equal to $2$, the exact solution can be found as well as the mean squared error rate.
So we randomly select $100$ unit vectors on the sphere as the vectors that we would like to estimate.
For both the Monte Carlo method and our optimal projection method, we calculate the mean squared error over these $100$ vectors.
More specifically, the squared error between the true norm of the vector, which is $1$, and the estimated norm is calculated for each of the $100$ unit vectors when the number of directions is fixed.
By taking the mean of the $100$ squared errors from the previous step, we get the mean squared error for given number of directions.
The number of directions used in our simulation is from $2$ to $10000$.
Figure \ref{fig:dim2} shows the comparison between our method and Monte Carlo method regarding the logarithm of the mean squared error and the number of projection directions.
From the figure, we can see that our method performs better than the Monte Carlo, and the advantage becomes more obvious when the number of projection directions increases.
\begin{figure}[htbp]
\centering
\includegraphics[width=0.7\textwidth]{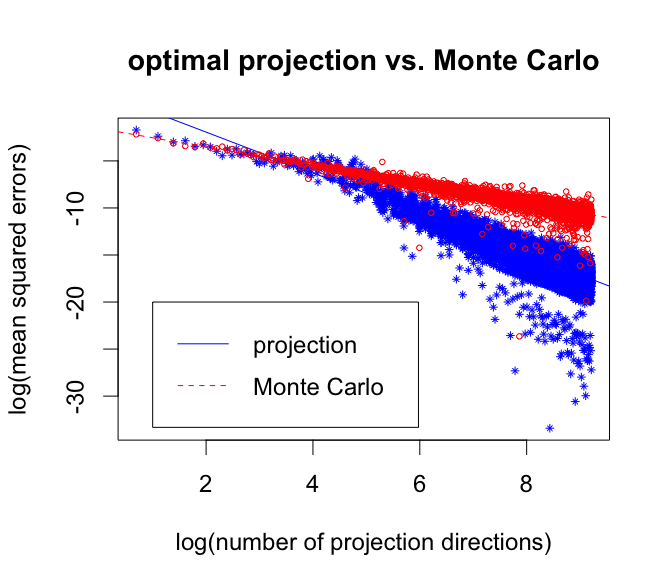}
\caption{Optimal projection vs. Monte Carlo in the 2 dimensional case}
\label{fig:dim2}
\end{figure}

\subsection{When we have $n=p$}
When the dimension $p$ is equal to the number of projection directions $n$, recall that in Theorem \ref{dim=n},
we give the exact solution of the pre-specified directions.
Similar to what we have done in the $2$-dimensional case, we randomly select 100 unit vectors on the sphere $S^{p-1}$,
with dimension $p$ varying from $8$ to $11$.
So the number of projection directions is varying from $8$ to $11$ correspondingly.
We calculate the mean squared error of both the Monte Carlo method and our optimal projection method for each $p$ using the same strategy as before.
The details are in the Figure \ref{fig:np}, where the $x$-axis represents the dimension, and $y$-axis represents the mean squared error.
\begin{figure}[htbp]
\centering
\includegraphics[width=0.7\textwidth]{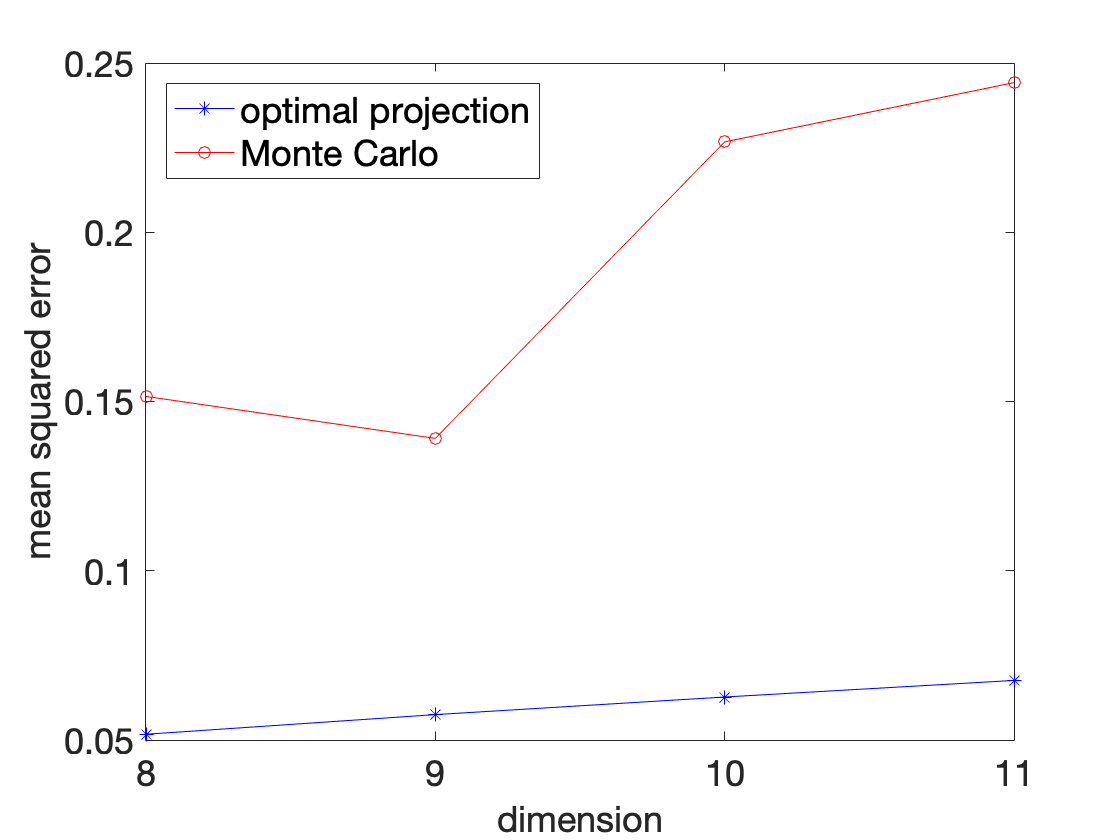}
\caption{Optimal projection vs. Monte Carlo in the $n=p$ case}
\label{fig:np}
\end{figure}

\subsection{General setting: $n > p$}
When the dimension $p$ is larger than 2 and $n\not=p$, the exact solution of \eqref{obj_f} can not be obtained.
Therefore, we adopt the Algorithm \ref{alg:general}.
Like in previous simulations, we randomly select 100 unit vectors on the sphere $S^{p-1}$, with dimension $p$ varying from $3$ to the number of directions minus $1$, and the fixed number of directions to be 8, 9, 10, 11, respectively, and calculate the mean squared error of both the Monte Carlo method and our optimal projection method for each $p$ using the same strategy as before.
Figure \ref{fig:n8}, \ref{fig:n9}, \ref{fig:n10} and \ref{fig:n11} show the comparison, where the $x$-axis represents the dimension, and $y$-axis represents the mean squared error.

\begin{figure}[htbp]
\centering
\includegraphics[width=0.7\textwidth]{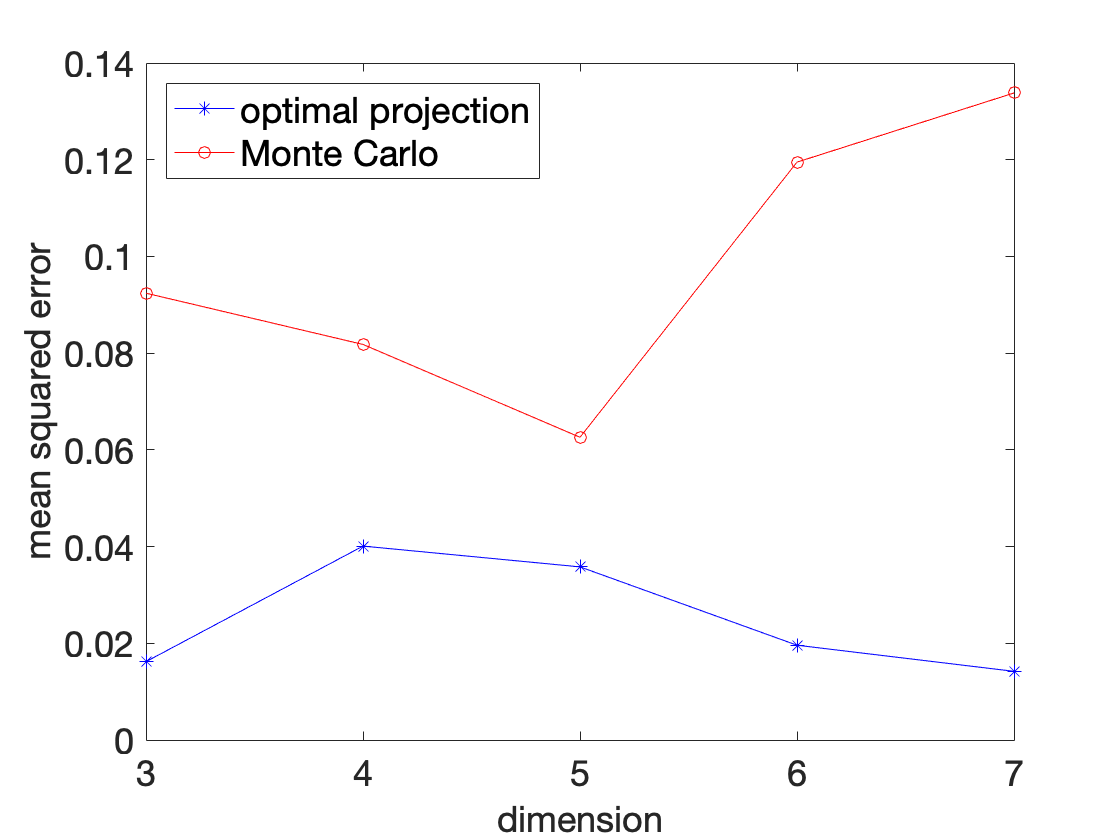}
\caption{Optimal projection vs. Monte Carlo for dimension varying from $3$ to $7$ in the case $n = 8$}
\label{fig:n8}
\end{figure}

\begin{figure}[htbp]
\centering
\includegraphics[width=0.7\textwidth]{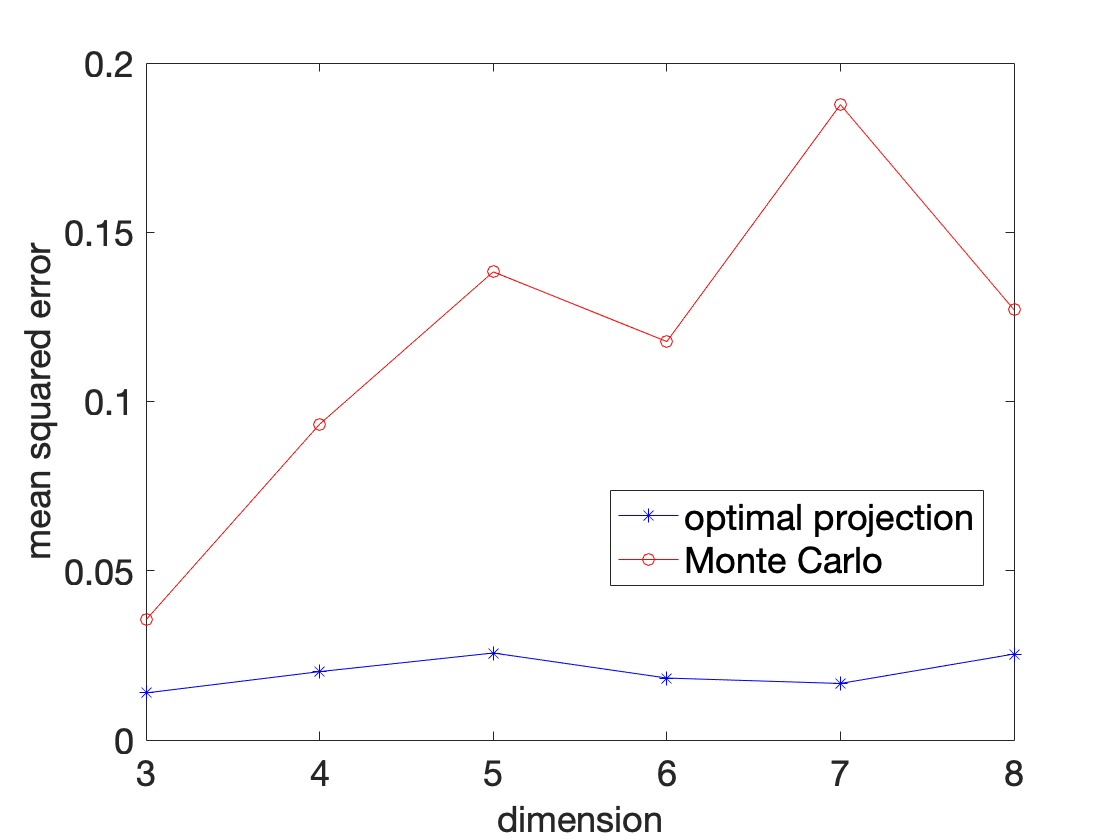}
\caption{Optimal projection vs. Monte Carlo for dimension varying from $3$ to $8$ in the case $n = 9$}
\label{fig:n9}
\end{figure}

\begin{figure}[htbp]
\centering
\includegraphics[width=0.7\textwidth]{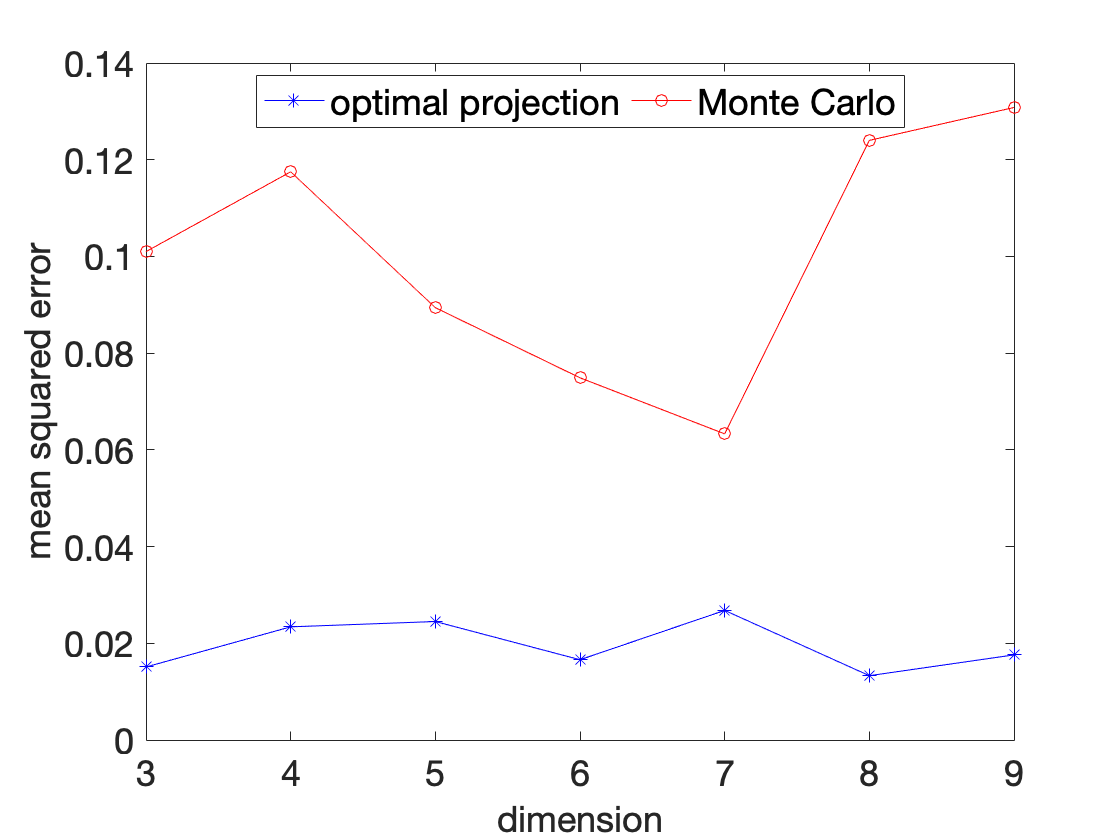}
\caption{Optimal projection vs. Monte Carlo for dimension varying from $3$ to $9$ in the case $n = 10$}
\label{fig:n10}
\end{figure}

\begin{figure}[htbp]
\centering
\includegraphics[width=0.7\textwidth]{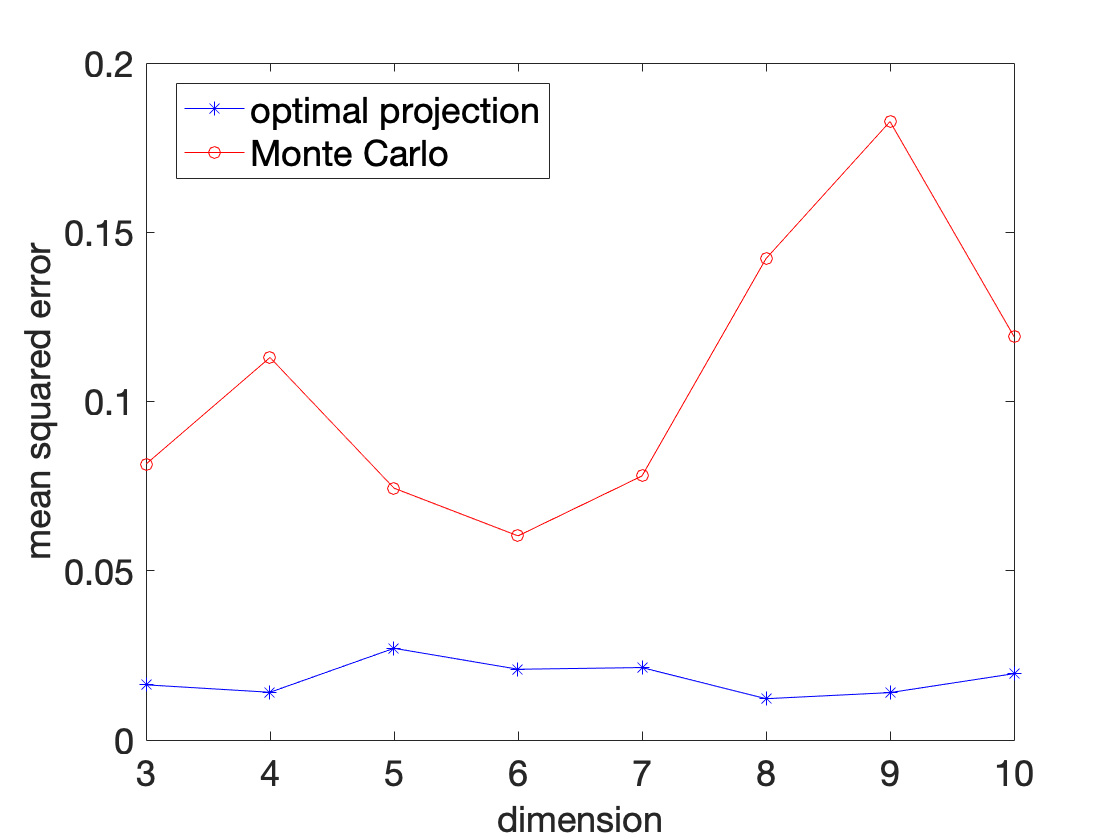}
\caption{Optimal projection vs. Monte Carlo for dimension varying from $3$ to $10$ in the case $n = 11$}
\label{fig:n11}
\end{figure}

Overall, we can see that our method performs better than the Monte Carlo method.


\section{Conclusion}
\label{sec:conc}

We propose a new method to calculate the distance, which is critical in computing the distance-based statistics, and can also be utilized in the calculation of the kernel functions that are distance-based and smooth.
The main idea is to use the sum of the norms of the projections along a set of pre-calculated directions to approximate the original norm.
By doing so, one can utilize the fast algorithm for univariate variables that has been proposed by \citet{huo2016fast}.
The advantage is that the computational complexity is reduced from $O(m^2)$ to $O(m \mbox{log}(m))$ where $m$ is the sample size.
These pre-specified directions can be found by minimizing the difference between the estimated distance and the true value in the worst case.
The associated problem is eventually a nonconvex optimization problem.
We derive the exact solutions when dimension is equal to either $2$ or the number of projection directions.
In general cases, we propose an algorithm to find the projection directions.
The simulations show the advantage of the proposed method versus the pure Monte Carlo approach, via comparing the mean squared errors.


\begin{acknowledgement}
This project is partially supported by the Transdisciplinary Research Institute for Advancing Data Science (TRIAD), http://
triad.gatech.edu, which is a part of the TRIPODS program at NSF and locates at Georgia Tech, enabled by the NSF grant
CCF-1740776.
Both authors are also partially supported by the NSF grant DMS-1613152.

\end{acknowledgement}
\section{Appendix}
\addcontentsline{toc}{section}{Appendix}
\label{sec:appendix}

All the proofs are included in this section, including
a proof of Theorem \ref{Cn} (Section \ref{sec:app:proof-1})),
a proof of Theorem \ref{dim2} (Section \ref{sec:app:proof-2}),
a proof of Theorem \ref{ese_O} (Section \ref{sec:app:proof-3}),
a proof of Theorem \ref{ese_M} (Section \ref{sec:app:proof-4}),
a proof of Lemma \ref{V_max} (Section \ref{sec:app:proof-5}),
a proof of Lemma \ref{prop} (Section \ref{sec:app:proof-6}),
a proof of Lemma \ref{Omega_0} (Section \ref{sec:app:proof-7}),
a proof of Theorem \ref{dim=n} (Section \ref{sec:app:proof-8}),
a proof of Lemma \ref{u_1} (Section \ref{sec:app:proof-9}), and
a proof of Theorem \ref{gtheta} (Section \ref{sec:app:proof-10}).
Some of these proofs involves detailed and potentially tedious derivations.
We try to furnish as much details as deemed reasonable.

\subsection{Proof of Theorem \ref{Cn}} \label{sec:app:proof-1}

\begin{proof}
By definition of $V_{\min}$ and $V_{\max}$, we have
$$C_nV_{\min}-1\leq C_n\sum_{i=1}^n\left|u_i^Tv\right|-1\leq C_nV_{\max}-1.$$
The above leads to the following
\begin{equation}\label{max}
\max\limits_{v: ||v||_2=1}\left|C_n\sum_{i=1}^n|u_i^Tv|-1\right|
= \max\left\{|C_nV_{\min}-1|, |C_nV_{\max}-1|\right\}.
\end{equation}
Consider the right hand side of the above as a function of $C_n$, it is verifiable that the minimum is achieved when
$$
1-C_nV_{\min}=C_nV_{\max}-1\mbox{, which leads to, }C_n = \frac{2}{V_{\min}+V_{\max}}.
$$
Bringing the above to \eqref{max}, we have
\begin{equation}\label{max_2}
\left|\frac{2}{V_{\min}+V_{\max}}V_{\min}-1\right|
= \frac{V_{\max}-V_{\min}}{V_{\max}+V_{\min}}
= \frac{2}{1+\frac{V_{\min}}{V_{\max}}}-1.
\end{equation}
From the above, it is evident that minimizing the right hand of \eqref{max_2} is equivalent to the following
$$
\max\limits_{u_1,...,u_n: ||u_i||_2=1}\frac{V_{\min}}{V_{\max}}.
$$
From all the above, the lemma is proved. \qed
\end{proof}


\subsection{Proof of Theorem \ref{dim2}}
\label{sec:app:proof-2}
\begin{proof}
Without loss of generality, we assume $\theta_i = \alpha_i + k_i\pi, \mbox{ where }\alpha_1\leq\alpha_2\leq ...\leq\alpha_n\in[0,\pi).$ Then the problem in \eqref{obj_f} can be written as
$$
\max\limits_{\alpha_i: i=1,...,n}\frac{\min\limits_\theta f(\theta)}{\max\limits_\theta f(\theta)},
$$
where $f(\theta) = \sum\limits_{i =1}^n|cos(\alpha_i-\theta)|.$

Let $\delta_i = \alpha_{i+1} - \alpha_i, i = 1, ..., n-1, \mbox{ and } \delta_n = \alpha_1 - \alpha_n + \pi.$
We have
$$
\sum\limits_{i=1}^n\delta_i = \pi.
$$
For given $\alpha_i$, the minimum and the maximum of $f(\theta)$ satisfy
\begin{eqnarray}
\frac{1}{n}\min\limits_\theta f(\theta) &\leq& \frac{1}{n}f(\alpha_i-\frac{\pi}{2}), \mbox{ for } i = 1, \dots, n, \label{eq:f1}\\
\frac{1}{n}\max\limits_\theta f(\theta) &\geq& \frac{1}{n}f\left(\frac{\alpha_i+\alpha_{i+1}}{2}-\frac{\pi}{2}\right), \mbox{ for } i = 1, \dots, n-1, \label{eq:f2-1}\\
\frac{1}{n}\max\limits_\theta f(\theta) &\geq& \frac{1}{n}f\left(\frac{\alpha_n+\alpha_1}{2}\right). \label{eq:f2-2}
\end{eqnarray}

By summing up each side of \eqref{eq:f1} with $i$ from $1$ through $n$, we get
\begin{equation}\label{eq:minf}
\min\limits_\theta f(\theta) \leq \frac{1}{n}\sum\limits_{i =1}^nf(\alpha_i-\frac{\pi}{2}).
\end{equation}
By summing up each side of \eqref{eq:f2-1} with $i$ from $1$ through $n-1$ and adding it to \eqref{eq:f2-2}, we have
\begin{equation}\label{eq:maxf}
\max\limits_\theta f(\theta) \geq \frac{1}{n}\left[\sum\limits_{i=1}^{n-1}f\left(\frac{\alpha_i+\alpha_{i+1}}{2}-\frac{\pi}{2}\right) + f\left(\frac{\alpha_n+\alpha_1}{2}\right)\right].
\end{equation}
Based on \eqref{eq:minf} and \eqref{eq:maxf}, for given $\alpha_i$, we have
$$
\frac{\min\limits_\theta f(\theta)}{\max\limits_\theta f(\theta)} \leq \frac{\frac{1}{n}\sum\limits_{i =1}^nf(\alpha_i-\frac{\pi}{2})}{\frac{1}{n}\left[\sum\limits_{i=1}^{n-1}f\left(\frac{\alpha_i+\alpha_{i+1}}{2}-\frac{\pi}{2}\right) + f\left(\frac{\alpha_n+\alpha_1}{2}\right)\right]}.
$$

Therefore, one can verify the following:
\begin{eqnarray}
\max\limits_{\alpha_i: i=1,...,n}\frac{\min\limits_\theta f(\theta)}{\max\limits_\theta f(\theta)}&\leq&
\max\limits_{\alpha_i: i=1,...,n}\frac{\frac{1}{n}\sum\limits_{i =1}^nf(\alpha_i-\frac{\pi}{2})}{\frac{1}{n}\left[\sum\limits_{i=1}^{n-1}f\left(\frac{\alpha_i+\alpha_{i+1}}{2}-\frac{\pi}{2}\right) + f\left(\frac{\alpha_n+\alpha_1}{2}\right)\right]} \nonumber \\
& = &\max\limits_{\alpha_i: i=1,...,n}\frac{\sum\limits_{i =1}^nf(\alpha_i-\frac{\pi}{2})}{\left[\sum\limits_{i=1}^{n-1}f\left(\frac{\alpha_i+\alpha_{i+1}}{2}-\frac{\pi}{2}\right) + f\left(\frac{\alpha_n+\alpha_1}{2}\right)\right]}. \label{less}
\end{eqnarray}

Denote the numerator of the right hand side of \eqref{less} as $N_n$, and the denominator as $D_n$.
Thus, we have
\begin{equation*}
N_n =
\left\{ \begin{array}{ll}
2\sum\limits_{i=1}^n\left|\sin\delta_i\right| + 2\sum\limits_{i=1}^{n-2}\left|\sin(\delta_i+\delta_{i+1})\right| + 2\sum\limits_{i=1}^{n-3}\left|\sin(\delta_i+\delta_{i+1} +\delta_{i+2})\right| + ... & \\
+ 2\sum\limits_{i=1}^{2}\left|\sin(\delta_i+\delta_{i+1} +... +\delta_{i+n-3})\right|, & \mbox{ if }n\geq 4,\\
2\sum\limits_{i=1}^n\left|\sin\delta_i\right| & \mbox{ if }n=3;
\end{array}
\right.
\end{equation*}
and
\begin{eqnarray*}
D_n &= &\sum\limits_{i=1}^n2\left|\sin\frac{\delta_i}{2}\right| + \sum\limits_{i=2}^{n-1}\sum\limits_{j=1}^{i-1}\left|\sin\left(\frac{\delta_i}{2} + \delta_{i-1} +\delta_{i-2} +...+\delta_j\right)\right| \\
&&+ \sum\limits_{i=1}^{n-2}\sum\limits_{j=i+1}^{n-1}\left|\sin\left(\frac{\delta_i}{2}+\delta_{i+1} + ... +\delta_{j}\right)\right|
+ \sum\limits_{j=2}^{n-1}\left|\sin\left(\frac{\delta_n}{2}+\delta_1 + ... +\delta_{j-1}\right)\right|.
\end{eqnarray*}

We would like to show that when all the $\theta_i$\rq s satisfies \eqref{optdim2}, $\frac{\min\limits_\theta f(\theta)}{\max\limits_\theta f(\theta)}$ is equal to the right hand side of \eqref{less}, which means \eqref{optdim2} is the optimal solution.
In order to do that, we first need to figure out what value the right hand side of \eqref{less} is.
In the following we use perturbation analysis to show that when $\delta_i = \frac{\pi}{n},$ which is equivalent to \eqref{optdim2}, the right hand side achieves the maximum value.
And then we show that the left side is equal to the right side under the condition of \eqref{optdim2}.
Therefore our proof can be completed.

For $n\geq4,$ $N_n$ and $D_n$ are treated as functions of $\Delta$.
Then we have
\begin{eqnarray*}
N_n(\delta_1+\Delta,\delta_2-\Delta,\delta_3,...,\delta_n) &=&
2\left|\sin(\delta_1+\Delta)\right| + 2\left|\sin(\delta_2-\Delta)\right| \\
&&+2\sum\limits_{j=3}^{n-1}\left|\sin(-\Delta+\sum\limits_{i=2}^j\delta_i)\right| + Const,
\end{eqnarray*}
and
\begin{eqnarray*}
\frac{\partial N_n(\delta_1+\Delta,\delta_2-\Delta,\delta_3,...,\delta_n)}{\partial \Delta}\bigg|_{\Delta=0}&=&
2\cos\delta_1\mbox{sign}(\sin\delta_1)-2\cos\delta_2\mbox{sign}(\sin\delta_2) \\
&&-2\sum\limits_{j=3}^{n-1}\cos\left(\sum\limits_{i=2}^j\delta_i\right)
\mbox{sign}\left(\sin\left(\sum\limits_{i=2}^j\delta_i\right)\right).
\end{eqnarray*}
When $\delta_i = \frac{\pi}{n}, i = 1,...,n,$ we have
\begin{eqnarray}
\frac{\partial N_n(\delta_1+\Delta,\delta_2-\Delta,\delta_3,...,\delta_n)}{\partial \Delta}\bigg|_{\Delta=0}&=&
0-2\sum\limits_{j=3}^{n-1}\cos\left(\frac{(j-1)\pi}{n}\right)\mbox{sign}\left(\sin\left(\frac{(j-1)\pi}{n}\right)\right) \nonumber \\
&=&0. \label{eq:Nn}
\end{eqnarray}

Similarly, for $D_n,$ we have
\begin{eqnarray*}
&&D_n(\delta_1+\Delta,\delta_2-\Delta,\delta_3,...,\delta_n) \\
&=& 2\left|\sin\left(\frac{\delta_1+\Delta}{2}\right)\right| + 2\left|\sin\left(\frac{\delta_2-\Delta}{2}\right)\right|
+ \left|\sin\left(\frac{\Delta}{2} + \delta_1 + \frac{\delta_2}{2}\right)\right|\\
&&+\sum\limits_{j=3}^{n-1}\left|\sin\left(-\Delta+\sum\limits_{i=2}^{j-1}\delta_i + \frac{\delta_j}{2}\right)\right|
+\sum\limits_{j=3}^{n}\left|\sin\left(-\frac{1}{2}\Delta+ \frac{\delta_1}{2}+\sum\limits_{i=2}^{j-1}\delta_i\right)\right|\\
&&+\sum\limits_{j=3}^{n-1}\left|\sin\left(-\frac{1}{2}\Delta+ \frac{\delta_2}{2}+\sum\limits_{i=3}^{j}\delta_i\right)\right|+\left|\sin\left(\frac{\delta_n}{2}+\delta_1+\Delta\right)\right| +Const,
\end{eqnarray*}
and
\begin{eqnarray*}
&&\frac{\partial D_n(\delta_1+\Delta,\delta_2-\Delta,\delta_3,...,\delta_n)}{\partial \Delta}\bigg|_{\Delta=0}\\
&=&\cos\frac{\delta_1}{2}\mbox{sign}(\sin\frac{\delta_1}{2})
-\cos\frac{\delta_2}{2}\mbox{sign}(\sin\frac{\delta_2}{2})
+\frac{1}{2}\cos\left(\frac{\delta_2}{2}+\delta_1\right)\mbox{sign}\left(\sin\left(\frac{\delta_2}{2}
+\delta_1\right)\right)\\
&&-\sum\limits_{j=3}^{n-1}\cos\left(\sum\limits_{i=2}^{j-1}\delta_i+ \frac{\delta_j}{2}\right)\mbox{sign}\left(\sin\left(\sum\limits_{i=2}^{j-1}\delta_i+ \frac{\delta_j}{2}\right)\right)\\
&&-\frac{1}{2}\sum\limits_{j=2}^{n-1}\cos\left(\frac{\delta_1}{2} + \sum\limits_{i =2}^{j}\delta_i\right)\mbox{sign}\left(\sin\left(\frac{\delta_1}{2}
 + \sum\limits_{i =2}^{j}\delta_i\right)\right)\\
&&-\frac{1}{2}\sum\limits_{j=3}^{n-1}\cos\left(\frac{\delta_2}{2} + \sum\limits_{i =3}^{j}\delta_i\right)\mbox{sign}\left(\sin\left(\frac{\delta_2}{2}
+ \sum\limits_{i =3}^{j}\delta_i\right)\right)\\
&&+\cos\left(\frac{\delta_n}{2}+\delta_1\right)\mbox{sign}\left(\sin\left(\frac{\delta_n}{2}+\delta_1\right)\right).
\end{eqnarray*}
When $\delta_i = \frac{\pi}{n}, i = 1,...,n,$ we have
\begin{eqnarray}
&&\frac{\partial D_n(\delta_1+\Delta,\delta_2-\Delta,\delta_3,...,\delta_n)}{\partial \Delta}\bigg|_{\Delta=0} \nonumber \\
&=&0+\frac{1}{2}\cos\left(\frac{3\pi}{2n}\right)\mbox{sign}\left(\sin\left(\frac{3\pi}{2n}\right)\right)
-\sum\limits_{j=3}^{n-1}\cos\left(\frac{(2j-3)\pi}{2n}\right)\mbox{sign}\left(\sin\left(\frac{(2j-3)\pi}{2n}\right)\right) \nonumber \\
&&-\frac{1}{2}\sum\limits_{j=2}^{n-1}\cos\left(\frac{(2j-1)\pi}{2n}\right)
\mbox{sign}\left(\sin\left(\frac{(2j-1)\pi}{2n}\right)\right) \nonumber \\
&&-\frac{1}{2}\sum\limits_{j=3}^{n-1}\cos\left(\frac{(2j-3)\pi}{2n}\right)
\mbox{sign}\left(\sin\left(\frac{(2j-3)\pi}{2n}\right)\right) \nonumber \\
&&+\cos\left(\frac{3\pi}{2n}\right)\mbox{sign}\left(\sin\left(\frac{3\pi}{2n}\right)\right) \nonumber \\
&=&0. \label{eq:Dn}
\end{eqnarray}

Define $g(\Delta)$ as the following
$$
g(\Delta) = \frac{N_n(\delta_1+\Delta,\delta_2-\Delta,\delta_3,...,\delta_n)}{D_n(\delta_1+\Delta,\delta_2-\Delta,\delta_3,...,\delta_n)}.
$$
Then we have
$$
\frac{\partial g(\Delta)}{\partial \Delta}\bigg |_{\Delta = 0} = \frac{N_n\rq\bigg|_{\Delta = 0}}{D_n(0)}-\frac{N_n(0)D_n\rq\bigg|_{\Delta = 0}}{D_n(0)^2},
$$
where
\begin{eqnarray*}
N_n\rq\bigg|_{\Delta = 0} &=& \frac{\partial N_n(\delta_1+\Delta,\delta_2-\Delta,\delta_3,...,\delta_n)}{\partial \Delta}\bigg|_{\Delta=0}, \\
D_n\rq\bigg|_{\Delta = 0} &=& \frac{\partial D_n(\delta_1+\Delta,\delta_2-\Delta,\delta_3,...,\delta_n)}{\partial \Delta}\bigg|_{\Delta=0};\\
N_n(0) &=& N_n(\delta_1,\delta_2,\delta_3,...,\delta_n), \\
D_n(0) &=& D_n(\delta_1,\delta_2,\delta_3,...,\delta_n).
\end{eqnarray*}
According to \eqref{eq:Nn} and \eqref{eq:Dn}, we have
$$
N_n\rq\bigg|_{\Delta = 0} = D_n\rq\bigg|_{\Delta = 0}  = 0.
$$
So we can get $\frac{\partial g(\Delta)}{\partial \Delta}\bigg |_{\Delta = 0} = 0-0=0.$

Similarly, for any two $\delta_i, \delta_j,$ simply give some perturbation to them, we can get the same result as above.
Therefore we can conclude that, for $n\geq 4$, $\left\{\delta_i = \frac{\pi}{n}, i = 1,...,n\right\}$ can maximize the function $\frac{N_n}{D_n}.$
Furthermore, we can get the maximum of  $\frac{N_n}{D_n}$ by letting each $\delta_i$ be $\frac{\pi}{n}$:
\begin{equation}\label{eq:NnDn}
\left(\frac{N_n}{D_n}\right)_{\max} =
\frac{2n\sin\frac{\pi}{n} + 2\sum\limits_{r=2}^{n-2}(n-r)\sin\frac{r\pi}{n}}{2n\sin\frac{\pi}{n}
+\sum\limits_{r=1}^{n-2}\left[2(n-r)-1\right]\sin\frac{(2r+1)\pi}{2n}}.
\end{equation}

Next, we would like to show that when $\delta_i = \frac{\pi}{n}, i = 1,...,n$, we have
$$
\frac{\min\limits_\theta f(\theta)}{\max\limits_\theta f(\theta)} = \left(\frac{N_n}{D_n}\right)_{\max}.
$$
As $f(\theta) = \sum\limits_{i =1}^n \left|\cos\left(\theta - \frac{(i-1)\pi}{n}\right)\right|$,
we know $f(\theta)=f\left(\theta - \frac{\pi}{n}\right)$.
So we only need to consider $\theta\in[0,\frac{\pi}{n}]$ to get the maximum.

Recall $f(\theta)$ is linear, so the minimum and maximum must be either $\theta=0$ or $\theta=\frac{\pi}{n}$.
By observing the periodicity of the function $f(\theta)$, we can get
\begin{eqnarray}
\min\limits_{\theta}f(\theta) &=&
 \begin{cases}
f(0) = 2\sum\limits_{r=1}^{a-1}\sin\frac{r\pi}{2a}+1&\mbox{ if $n=2a$,}\\\nonumber
f\left(\frac{\pi}{2(2a+1)}\right) = 2\sum\limits_{r=1}^{a}\sin\frac{r\pi}{2a+1}&\mbox{ if $n=2a+1$.} \\
 \end{cases}\\
\max\limits_{\theta}f(\theta) &=&
 \begin{cases}
f\left(\frac{\pi}{4a}\right) = 2\sum\limits_{r=1}^{a}\sin\frac{(2r-1)\pi}{4a}&\mbox{ if $n=2a$.} \\
f(0) = 2\sum\limits_{r=1}^{a}\sin\frac{(2r-1)\pi}{2(2a+1)}+1&\mbox{ if $n=2a+1$,}\nonumber
 \end{cases}
\end{eqnarray}

From \eqref{eq:NnDn} we can get
\begin{eqnarray}
\label{eq:fmin/fmax}
\left(\frac{N_n}{D_n}\right)_{\max} =
 \begin{cases}
\frac{2\sum\limits_{r=1}^{a-1}\sin\frac{r\pi}{2a}+1}{2\sum\limits_{r=1}^{a}\sin\frac{(2r-1)\pi}{4a}}&\mbox{ if $n=2a$.} \\
\frac{2\sum\limits_{r=1}^{a}\sin\frac{r\pi}{2a+1}}{2\sum\limits_{r=1}^{a}\sin\frac{(2r-1)\pi}{2(2a+1)}+1}&\mbox{ if $n=2a+1$,} \\
 \end{cases}
\end{eqnarray}
Therefore, we can conclude that when $\delta_i = \frac{\pi}{n}, i = 1,...,n$,
$$
\frac{\min\limits_\theta f(\theta)}{\max\limits_\theta f(\theta)} = \left(\frac{N_n}{D_n}\right)_{\max}.
$$
Recall the definition of $\delta_i$\rq s, we know that \eqref{optdim2} is the optimal solution for $n\geq4.$

For $n=3\mbox{ and }2,$ by applying the similar strategy, we can get the same result as above.
 \qed
\end{proof}

\subsection{Propositions we need in order to prove Theorem \ref{ese_O}}
Before proceeding to the proof of Theorem \ref{ese_O}, we need the following Proposition \ref{prop_2} and \ref{prop_1}:

\begin{proposition}
\label{prop_2}
\begin{eqnarray*}
&&\sum\limits_{s=1}^{n-1}\sin\frac{s}{n}\pi=\mbox{cot}\frac{\pi}{2n},\\
&&\sum\limits_{s=1}^{n-1}\cos\frac{s}{n}\pi=0,\\
&&\sum\limits_{s=1}^{n-1}s\sin\frac{s}{n}\pi=\frac{n}{2}\mbox{cot}\frac{\pi}{2n},\\
&&\sum\limits_{s=1}^{n-1}s\cos\frac{s}{n}\pi=-\frac{1}{2}\mbox{cot}^2\frac{\pi}{2n}+\frac{n-1}{2},\\
&&\sum\limits_{s=1}^{n-1}s^2\cos\frac{s}{n}\pi=-\frac{n}{2}\mbox{cot}^2\frac{\pi}{2n}+\frac{n(n-1)}{2}.
\end{eqnarray*}
\end{proposition}

\begin{proof}
As the following holds true
$$
\sin\frac{s\pi}{N}\sin\frac{\pi}{2n}
=\frac{1}{2}\left(\cos\frac{(2s-1)\pi}{2n}-\cos\frac{(2s+1)\pi}{2n}\right),
$$
we have
\begin{eqnarray*}
&&\left(\sum\limits_{s=1}^{n-1}\sin\frac{s}{n}\pi\right)\cdot\sin\frac{\pi}{2n} \\
&=&\frac{1}{2}\sum\limits_{s=1}^{n-1}\left(\cos\frac{(2s-1)\pi}{2n}-\cos\frac{(2s+1)\pi}{2n}\right)
=\frac{1}{2}\left(\cos\frac{\pi}{2n}-\cos\frac{(2n-1)\pi}{2n}\right)=\cos\frac{\pi}{2n}.
\end{eqnarray*}
So by dividing $\sin\frac{\pi}{2n}$ for both sides, we can get
\begin{equation}
\label{sum_sin}
\sum\limits_{s=1}^{n-1}\sin\frac{s}{n}\pi=\mbox{cot}\frac{\pi}{2n}.
\end{equation}

As we also have
$$
\cos\frac{s\pi}{N}\sin\frac{\pi}{2n}
=\frac{1}{2}\left(\sin\frac{(2s+1)\pi}{2n}-\sin\frac{(2s-1)\pi}{2n}\right).
$$
Therefore, we can get
\begin{eqnarray*}
&&\left(\sum\limits_{s=1}^{n-1}\cos\frac{s}{n}\pi\right)\cdot\sin\frac{\pi}{2n} \\
&=&\frac{1}{2}\sum\limits_{s=1}^{n-1}\left(\sin\frac{(2s+1)\pi}{2n}-\sin\frac{(2s-1)\pi}{2n}\right)
=\frac{1}{2}\left(\sin\frac{(2n-1)\pi}{2n}-\sin\frac{\pi}{2n}\right)=0,
\end{eqnarray*}
which implies
\begin{equation}
\label{sum_cos}
\sum\limits_{s=1}^{n-1}\cos\frac{s}{n}\pi=0.
\end{equation}

As we also have
$$
\sin\frac{s}{n}\pi\cdot\sin\frac{\pi}{2n} = \cos\frac{(2s-1)\pi}{2n}-\cos\frac{(2s+1)\pi}{2n},
$$
the following can be derived:
\begin{equation}\label{eq:sumsinsin}
\left(\sum\limits_{s=1}^{n-1}s\sin\frac{s}{n}\pi\right)\cdot\sin\frac{\pi}{2n}
= \frac{1}{2}\sum\limits_{s=1}^{n-1}s\cdot\left(\cos\frac{(2s-1)\pi}{2n}-\cos\frac{(2s+1)\pi}{2n}\right).
\end{equation}
Since we have
\begin{eqnarray*}
&&\sum\limits_{s=1}^{n-1}s\cdot\left(\cos\frac{(2s-1)\pi}{2n}-\cos\frac{(2s+1)\pi}{2n}\right)\\
&=& \sum\limits_{s=1}^{n-1}\cos\frac{(2s-1)\pi}{2n}-(n-1)\cos\frac{(2n-1)\pi}{2n} \\
&=& \sum\limits_{s=1}^{n-1}\left(\cos\frac{s}{n}\pi\cos\frac{\pi}{2n}
+\sin\frac{s}{n}\pi\sin\frac{\pi}{2n}\right)+(n-1)\cos\frac{\pi}{2n},
\end{eqnarray*}
by plugging the above as well as \eqref{sum_sin} and \eqref{sum_cos} into \eqref{eq:sumsinsin}, we can get
\begin{eqnarray}
&&\left(\sum\limits_{s=1}^{n-1}s\sin\frac{s}{n}\pi\right)\cdot\sin\frac{\pi}{2n} \nonumber \\
&=&\frac{1}{2}\left(\sum\limits_{s=1}^{n-1}\left(\cos\frac{s}{n}\pi\cos\frac{\pi}{2n}
+\sin\frac{s}{n}\pi\sin\frac{\pi}{2n}\right)+(n-1)\cos\frac{\pi}{2n}\right) \nonumber \\
&=&\frac{1}{2}\left(0+\cos\frac{\pi}{2n}\right)+\frac{n-1}{2}\cos\frac{\pi}{2n}
=\frac{n}{2}\cos\frac{\pi}{2n}. \label{sum_times_sin}
\end{eqnarray}
Similarly, since we have
$$
\cos\frac{s}{n}\pi\cdot\sin\frac{\pi}{2n} = \sin\frac{(2s+1)\pi}{2n}-\sin\frac{(2s-1)\pi}{2n},
$$
by using the similar strategy, we can get
\begin{equation}
\left(\sum\limits_{s=1}^{n-1}s\cos\frac{s}{n}\pi\right)\cdot\sin\frac{\pi}{2n} 
=-\frac{1}{2}\cos\frac{\pi}{2n}\mbox{cot}\frac{\pi}{2n}+\frac{n-1}{2}\sin\frac{\pi}{2n}. \label{sum_times_sin_2}
\end{equation}

Therefore, dividing both the equations \eqref{sum_times_sin} and \eqref{sum_times_sin_2}, we can get
\begin{eqnarray}
&&\sum\limits_{s=1}^{n-1}s\sin\frac{s}{n}\pi=\frac{n}{2}\mbox{cot}\frac{\pi}{2n}, \label{eq:sum_ssin}\\
&&\sum\limits_{s=1}^{n-1}s\cos\frac{s}{n}\pi=-\frac{1}{2}\mbox{cot}^2\frac{\pi}{2n}+\frac{n-1}{2}. \label{eq:sum_scos}
\end{eqnarray}

Since the following holds true,
$$
\left(\sum\limits_{s=1}^{n-1}s^2\cos\frac{s}{n}\pi\right)\cdot\sin\frac{\pi}{2n}
=\frac{1}{2}\left\{-\sum\limits_{s=1}^{n-1}(2s-1)\sin\frac{(2s-1)\pi}{2n}+(n-1)^2\sin\frac{2n-1}{2n}\pi\right\},
$$
by simplifying the above equation we can get
\begin{eqnarray*}
&&\left(\sum\limits_{s=1}^{n-1}s^2\cos\frac{s}{n}\pi\right)\cdot\sin\frac{\pi}{2n} \\
&=&-\sum\limits_{s=1}^{n-1}s\left(\sin\frac{s}{n}\pi\cos\frac{\pi}{2n}
-\cos\frac{s}{n}\pi\sin\frac{\pi}{2n}\right)
+\frac{1}{2}\sum\limits_{s=1}^{n-1}\left(\sin\frac{s}{n}\pi\cos\frac{\pi}{2n}
-\cos\frac{s}{n}\pi\sin\frac{\pi}{2n}\right) \\
&&+\frac{(n-1)^2}{2}\sin\frac{\pi}{2n} \\
&=&-\left(\sum\limits_{s=1}^{n-1}s\sin\frac{s}{n}\pi\right)\cos\frac{\pi}{2n}
+ \left(\sum\limits_{s=1}^{n-1}s\cos\frac{s}{n}\pi\right)\sin\frac{\pi}{2n}
+\frac{1}{2}\left(\sum\limits_{s=1}^{n-1}\sin\frac{s}{n}\pi\right)\cos\frac{\pi}{2n} \\
&&-\frac{1}{2}\left(\sum\limits_{s=1}^{n-1}\cos\frac{s}{n}\pi\right)\sin\frac{\pi}{2n}
+\frac{(n-1)^2}{2}\sin\frac{\pi}{2n}. \\
\end{eqnarray*}
Plugging \eqref{sum_sin}, \eqref{sum_cos}, \eqref{eq:sum_ssin}, and \eqref{eq:sum_scos} into the above, we can get
$$
\left(\sum\limits_{s=1}^{n-1}s^2\cos\frac{s}{n}\pi\right)\cdot\sin\frac{\pi}{2n}
=-\frac{1}{2}\mbox{cot}^2\frac{\pi}{2n}\sin\frac{\pi}{2n}
+\frac{n(n-1)}{2}\sin\frac{\pi}{2n}.
$$
Therefore, dividing $\sin\frac{\pi}{2n}$ on each side, we get
$$
\sum\limits_{s=1}^{n-1}s^2\cos\frac{s}{n}\pi
=-\frac{n}{2}\mbox{cot}^2\frac{\pi}{2n}+\frac{n(n-1)}{2}.
$$ \qed
\end{proof}


\begin{proposition}
\label{prop_1}
$$
2\sum\limits_{s=1}^{n-1}(n-s)f(s)
=\frac{n}{\pi}\mbox{cot}\frac{\pi}{2n}+\frac{1}{2}\mbox{cot}^2\frac{\pi}{2n}-\frac{n}{2}+\frac{1}{2}.
$$
\end{proposition}

\begin{proof}
According to the definition of function $f(s)$ in \eqref{eq:fs}, we have
\begin{eqnarray*}
&&2\sum\limits_{s=1}^{n-1}(n-s)f(s) \\
&=&2\sum\limits_{s=1}^{n-1}(n-s)\left(\frac{1}{\pi}\sin\frac{s}{n}\pi+\left(\frac{1}{2}-\frac{s}{n}\right)\cos\frac{s}{n}\pi\right) \\
&=&\frac{2n}{\pi}\sum\limits_{s=1}^{n-1}\sin\frac{s}{n}\pi-\frac{2}{\pi}\sum\limits_{s=1}^{n-1}s\sin\frac{s}{n}\pi
+n\sum\limits_{s=1}^{n-1}\cos\frac{s}{n}\pi-3\sum\limits_{s=1}^{n-1}s\cos\frac{s}{n}\pi
+\frac{2}{n}\sum\limits_{s=1}^{n-1}s^2\cos\frac{s}{n}\pi.
\end{eqnarray*}
Applying Proposition \ref{prop_2},we have
$$
2\sum\limits_{s=1}^{n-1}(n-s)f(s)
=\frac{n}{\pi}\mbox{cot}\frac{\pi}{2n}+\frac{1}{2}\mbox{cot}^2\frac{\pi}{2n}-\frac{n}{2}+\frac{1}{2}.  \qquad \qquad \qed
$$
\end{proof}


\subsection{Proof of Theorem \ref{ese_O}}
\label{sec:app:proof-3}
\begin{proof}
Recall that $u_i$ can be rewritten as $$u_i = e^{\sqrt{-1} \frac{i\pi}{n}}, i = 0,1,\dots,n-1.$$
And we have
\begin{eqnarray}
&&\mathop{\mathbb{E}}\limits_{v\sim \mbox{Unif}(S^1)}\left\{\left|C_n\sum\limits_{i=1}^n\left|u_i^Tv\right|-1\right|^2\right\} \nonumber \\
&=&C_n^2\mathop{\mathbb{E}}\limits_{v\sim \mbox{Unif}(S^1)}\left\{\left(\sum\limits_{i=1}^n\left|u_i^Tv\right|\right)^2\right\}-2C_n\mathop{\mathbb{E}}\limits_{v\sim \mbox{Unif}(S^1)}\left\{\sum\limits_{i=1}^n\left|u_i^Tv\right|\right\}+1 \nonumber \\
&=&C_n^2\sum\limits_{i=1}^n\mathop{\mathbb{E}}\limits_{v\sim \mbox{Unif}(S^1)}\left(\left|u_i^Tv\right|^2\right)+2C_n^2\sum\limits_{1\leq i<j\leq N}\mathop{\mathbb{E}}\limits_{v\sim \mbox{Unif}(S^1)}\left(\left|u_i^Tv\right|\left|u_j^Tv\right|\right) \nonumber \\
&&-2C_n\sum\limits_{i=1}^n\mathop{\mathbb{E}}\limits_{v\sim \mbox{Unif}(S^1)}\left\{\left|u_i^Tv\right|\right\}+1 \label{eq_0}.
\end{eqnarray}

So we will find out the expected squared error, if for all $i,j = 1, \dots, n$, we can get the values of
$$
\mathop{\mathbb{E}}\limits_{v\sim \mbox{Unif}(S^1)}\left(\left|u_i^Tv\right|^2\right), ~~~
\mathop{\mathbb{E}}\limits_{v\sim \mbox{Unif}(S^1)}\left(\left|u_i^Tv\right|\left|u_j^Tv\right|\right), ~~~
\mathop{\mathbb{E}}\limits_{v\sim \mbox{Unif}(S^1)}\left\{\left|u_i^Tv\right|\right\}.
$$

In order to calculate $\mathop{\mathbb{E}}\limits_{v\sim \mbox{Unif}(S^1)}\left(\left|u_i^Tv\right|^2\right)$, we let $u_i = (1,0)^\prime$ and $v = (\cos\theta, \sin\theta)^\prime$ without loss of generality.
Then,
$$
\mathop{\mathbb{E}}\limits_{v\sim \mbox{Unif}(S^1)}\left(\left|u_i^Tv\right|^2\right)
= \mathop{\mathbb{E}}\limits_{\theta\sim \mbox{Unif}(0,2\pi)}\cos^2\theta
= \frac{1}{2}+\frac{1}{2}\mathop{\mathbb{E}}\limits_{\theta\sim \mbox{Unif}(0,2\pi)}\cos2\theta
= \frac{1}{2}.
$$

Without loss of generality, assume $\langle u_i, u_j\rangle= \frac{s}{n}\pi, \mbox{ for all }1\leq i,j\leq n, i\not=j, $ which means we can assume
$$
u_i = (1,0)^\prime, u_j = (\cos\frac{s}{n}\pi, \sin\frac{s}{n}\pi)^\prime, s = 1,2,\dots, n-1.
$$
Therefore, we have
\begin{eqnarray*}
\left|u_i^Tv\right|\cdot\left|u_j^Tv\right| &=& \left|\cos\theta\right|\left|\cos\theta\cos\frac{s}{n}\pi
+\sin\theta\sin\frac{s}{n}\pi\right| \\
&=& \left|\cos^2\theta\cos\frac{s}{n}\pi
+ \cos\theta\sin\theta\sin\frac{s}{n}\pi\right|.
\end{eqnarray*}
As the following equations hold,
$$
\cos^2\theta=\frac{1+\cos2\theta}{2} \mbox{ and } \cos\theta\sin\theta = \frac{\sin2\theta}{2},
$$
quantity $\left|u_i^Tv\right|\cdot\left|u_j^Tv\right|$ can be further written as
\begin{eqnarray*}
\left|u_i^Tv\right|\cdot\left|u_j^Tv\right|&=&\frac{1}{2}\left|\cos2\theta\cos\frac{s}{n}\pi
+\sin2\theta\sin\frac{s}{n}\pi+\cos\frac{s}{n}\pi\right| \\
&=& \frac{1}{2}\left|\cos\left(2\theta-\frac{s}{n}\pi\right)
+\cos\frac{s}{n}\pi\right|.
\end{eqnarray*}
So $\mathop{\mathbb{E}}\limits_{v\sim \mbox{Unif}(S^1)}\left(\left|u_i^Tv\right|\cdot\left|u_j^Tv\right|\right)$ can be rewritten as follows:
\begin{eqnarray*}
&&\mathop{\mathbb{E}}\limits_{v\sim \mbox{Unif}(S^1)}\left(\left|u_i^Tv\right|\cdot\left|u_j^Tv\right|\right) \\
&=&\frac{1}{2}\mathop{\mathbb{E}}\limits_{\theta\sim \mbox{Unif}(0,2\pi)}\left\{\left|\cos\left(2\theta-\frac{s}{n}\pi\right)
+\cos\frac{s}{n}\pi\right|\right\}\\
&=&\frac{1}{2}\times\frac{1}{2\pi}\left(\int\limits_0^\pi
+\int\limits_\pi^{2\pi}\right)\left|\cos\left(2\theta-\frac{s}{n}\pi\right)+\cos\frac{s}{n}\pi\right|\mbox{d}\theta.\\
\end{eqnarray*}
As we have
\begin{eqnarray*}
&&\int\limits_\pi^{2\pi}\left|\cos\left(2\theta-\frac{s}{n}\pi\right)+\cos\frac{s}{n}\pi\right|\mbox{d}\theta \\
&=& \int\limits_0^\pi\left|\cos\left(2\theta-\frac{s}{n}\pi\right)+\cos\frac{s}{n}\pi\right|\mbox{d}\theta
= \int\limits_0^\pi\left|\cos\left(2\theta\right)+\cos\frac{s}{n}\pi\right|d\theta \\
&=&\int_{-\frac{\pi}{2}+\frac{s}{2n}\pi}^{\frac{\pi}{2}+\frac{s}{2n}\pi}\left|\cos\left(2\theta\right)+\cos\frac{s}{n}\pi\right|d\theta,
\end{eqnarray*}
we can get
\begin{equation}\label{eq:Euv}
\mathop{\mathbb{E}}\limits_{v\sim \mbox{Unif}(S^1)}\left(\left|u_i^Tv\right|\cdot\left|u_j^Tv\right|\right)
=\frac{1}{2\pi}\int_{-\frac{\pi}{2}+\frac{s}{2n}\pi}^{\frac{\pi}{2}-\frac{s}{2n}\pi}\left|\cos\left(2\theta\right)+\cos\frac{s}{n}\pi\right|d\theta.
\end{equation}
By breaking the integral interval $(-\frac{\pi}{2}+\frac{s}{2n}\pi, \frac{\pi}{2}+\frac{s}{2n}\pi)$ into two subintervals,
$(-\frac{\pi}{2}+\frac{s}{2n}\pi, \frac{\pi}{2}-\frac{s}{2n}\pi)$ and
$(\frac{\pi}{2}-\frac{s}{2n}\pi, \frac{\pi}{2}+\frac{s}{2n}\pi)$, we have
\begin{eqnarray*}
\left|\cos2\theta+\cos\frac{s}{n}\pi\right| =
\begin{cases}
\cos2\theta+\cos\frac{s}{n}\pi, & \theta\in(-\frac{\pi}{2}+\frac{s}{2n}\pi, \frac{\pi}{2}-\frac{s}{2n}\pi),\\
-\left(\cos2\theta+\cos\frac{s}{n}\pi\right), &\theta\in(\frac{\pi}{2}-\frac{s}{2n}\pi, \frac{\pi}{2}+\frac{s}{2n}\pi).
\end{cases}
\end{eqnarray*}
Combining \eqref{eq:Euv}, we get
\begin{eqnarray*}
&&\mathop{\mathbb{E}}\limits_{v\sim \mbox{Unif}(S^1)}\left(\left|u_i^Tv\right|\cdot\left|u_j^Tv\right|\right)\\
&=&\frac{1}{2\pi}\left(\int_{-\frac{\pi}{2}+\frac{s}{2n}\pi}^{\frac{\pi}{2}-\frac{s}{2n}\pi}
+\int_{\frac{\pi}{2}-\frac{s}{2n}\pi}^{\frac{\pi}{2}+\frac{s}{2n}\pi}\right)\left|\cos2\theta+\cos\frac{s}{n}\pi\right|d\theta\\
&=&\frac{1}{2\pi}\left\{\int_{-\frac{\pi}{2}+\frac{s}{2n}\pi}^{\frac{\pi}{2}-\frac{s}{2n}\pi}\left(\cos2\theta
+\cos\frac{s}{n}\pi\right)d\theta-\int_{\frac{\pi}{2}-\frac{s}{2n}\pi}^{\frac{\pi}{2}+\frac{s}{2n}\pi}\left(\cos2\theta
+\cos\frac{s}{n}\pi\right)d\theta\right\} \\
&=&\frac{1}{2\pi}\left\{2\sin\frac{s}{n}\pi+\left(\pi-\frac{2s}{N}\pi\right)\cos\frac{s}{n}\pi\right\} \\
&=&\frac{1}{\pi}\sin\frac{s}{n}\pi+\left(\frac{1}{2}-\frac{s}{n}\right)\cos\frac{s}{n}\pi.
\end{eqnarray*}
If we define
\begin{equation}\label{eq:fs}
f(s)=\frac{1}{\pi}\sin\frac{s}{n}\pi+\left(\frac{1}{2}-\frac{s}{n}\right)\cos\frac{s}{n}\pi, s = 0,1,2,\cdots,n-1.
\end{equation}
Then we will get
\begin{eqnarray}
\mathop{\mathbb{E}}\limits_{v\sim \mbox{Unif}(S^1)}\left(\left|u_i^Tv\right|^2\right) &=& f(0), \nonumber \\
\mathop{\mathbb{E}}\limits_{v\sim \mbox{Unif}(S^1)}\left(\left|u_i^Tv\right|\cdot\left|u_j^Tv\right|\right) &=&f(s),
\mbox{ where }\langle u_i, u_j\rangle=\frac{s}{n}\pi, s = 1,2,\cdots,n-1. \label{eq_1}
\end{eqnarray}

Similarly, without loss of generality, if we assume
$u_i = (1,0)^\prime, v = (\cos\theta,\sin\theta)^\prime,$ the following holds,
\begin{equation}
\label{eq_2}
\mathop{\mathbb{E}}\limits_{v\sim \mbox{Unif}(S^1)}\left\{\left|u_i^Tv\right|\right\}
= \mathop{\mathbb{E}}\limits_{\theta\sim \mbox{Unif}(-\pi,\pi)}\left|\cos\theta\right|
= 2\int\limits_{-\frac{\pi}{2}}^{\frac{\pi}{2}}\frac{1}{2\pi}\cos\theta d\theta = \frac{2}{\pi}.
\end{equation}

Recall that we have
$$
C_n = \frac{2}{\mbox{V}_{\min}+\mbox{V}_{\max}},
\mbox{ where } \mbox{V}_{\min}=\min\limits_{v:\|v\|=1}\sum\limits_{i=1}^n\left|u_i^Tv\right|,
 \mbox{V}_{\max}=\max\limits_{v:\|v\|=1}\sum\limits_{i=1}^n\left|u_i^Tv\right|.
$$
From \eqref{eq:fmin/fmax} we can easily verify that
$$
\mbox{V}_{\min}+\mbox{V}_{\max}=2\sum\limits_{k=1}^{n-1}\sin\frac{k\pi}{2n} +1.
$$
Therefore, $C_n$ can be derived:
\begin{equation}\label{eq:Cn1}
C_n = \frac{2}{2\sum\limits_{k=1}^{n-1}\sin\frac{k\pi}{2n} +1}.
\end{equation}

As we have
$$
\sin\frac{k\pi}{2n}\cdot\sin\frac{\pi}{4n}
=\frac{1}{2}\left(\cos\frac{(2k-1)\pi}{4n}-\cos\frac{(2k+1)\pi}{4n}\right),
$$
we can get
$$
\sin\frac{\pi}{4n}\left(\sum\limits_{k=1}^{n-1}\sin\frac{k\pi}{2n}\right)
=\frac{1}{2}\sum\limits_{k=1}^{n-1}\left(\cos\frac{(2k-1)\pi}{4n}-\cos\frac{(2k+1)\pi}{4n}\right)
=\frac{1}{2}\left(\cos\frac{\pi}{4n}-\sin\frac{\pi}{4n}\right),
$$
which leads to
\begin{equation}
\label{eq:triangle1}
\sum\limits_{k=1}^{n-1}\sin\frac{k\pi}{2n}
= \frac{\frac{1}{2}\left(\cos\frac{\pi}{4n}-\sin\frac{\pi}{4n}\right)}{\sin\frac{\pi}{4n}}
=\frac{1}{2}\mbox{cot}\frac{\pi}{4n}-\frac{1}{2}.
\end{equation}
Therefore, by plugging \eqref{eq:triangle1} into \eqref{eq:Cn1}, we have
$$
C_n=\frac{2}{\mbox{cot}\frac{\pi}{4n}} = 2\mbox{tan}\frac{\pi}{4n}.
$$

If we plug in \eqref{eq_0} with \eqref{eq_1} and \eqref{eq_2}, we can get

\begin{eqnarray}
&&\mathop{\mathbb{E}}\limits_{v\sim \mbox{Unif}(S^1)}\left\{\left|C_n\sum\limits_{i=1}^n\left|u_i^Tv\right|-1\right|^2\right\} \nonumber \\
&=&C_n^2\sum\limits_{i=1}^n\mathop{\mathbb{E}}\limits_{v\sim \mbox{Unif}(S^1)}\left(\left|u_i^Tv\right|^2\right)+2C_n^2\sum\limits_{1\leq i<j\leq N}\mathop{\mathbb{E}}\limits_{v\sim \mbox{Unif}(S^1)}\left(\left|u_i^Tv\right|\left|u_j^Tv\right|\right) \nonumber \\
&&-2C_n\sum\limits_{i=1}^n\mathop{\mathbb{E}}\limits_{v\sim \mbox{Unif}(S^1)}\left\{\left|u_i^Tv\right|\right\}+1 \nonumber \\
&=&4\mbox{tan}^2\frac{\pi}{4n}\left(\frac{n}{2}+2\sum\limits_{s=1}^{n-1}(n-s)f(s)\right)-\frac{8N}{\pi}\mbox{tan}\frac{\pi}{4n} +1. \label{eq_3}
\end{eqnarray}
In order to calculate the part $\sum\limits_{s=1}^{n-1}(n-s)f(s)$ in \eqref{eq_3}, we need the Proposition \ref{prop_1}.
Applying Proposition \ref{prop_1} on \eqref{eq_3}, we get
\begin{eqnarray}
&&\mathop{\mathbb{E}}\limits_{v\sim \mbox{Unif}(S^1)}\left\{\left|C_n\sum\limits_{i=1}^n\left|u_i^Tv\right|-1\right|^2\right\} \nonumber \\
&=&4\mbox{tan}^2\frac{\pi}{4n}\left(\frac{n}{2}+\frac{n}{\pi}\mbox{cot}\frac{\pi}{2n}+\frac{1}{2}\mbox{cot}^2\frac{\pi}{2n}-\frac{n}{2}
+\frac{1}{2}\right)-\frac{8n}{\pi}\mbox{tan}\frac{\pi}{4n} +1  \nonumber \\
&=&2\mbox{tan}^2\frac{\pi}{4n}\mbox{cot}^2\frac{\pi}{2n}+\frac{4n}{\pi}\mbox{tan}^2\frac{\pi}{4n}\mbox{cot}\frac{\pi}{2n}
+2\mbox{tan}^2\frac{\pi}{4n}-\frac{8n}{\pi}\mbox{tan}\frac{\pi}{4n} +1. \label{eq_4}
\end{eqnarray}


As $\mbox{tan}x\rightarrow x,\mbox{ as } x\rightarrow 0$, we can get
\begin{eqnarray*}
\mathop{\mathbb{E}}\limits_{v\sim \mbox{Unif}(S^1)}\left\{\left|C_n\sum\limits_{i=1}^n\left|u_i^Tv\right|-1\right|^2\right\}
&\longrightarrow & 2\frac{\pi^2}{16n^2}\frac{4n^2}{\pi^2}+\frac{4n}{\pi}\frac{\pi^2}{16n^2}\frac{2n}{\pi}+2\frac{\pi^2}{16n^2}
-\frac{8n}{\pi}\frac{\pi}{4n}+1 \\
&=&\frac{\pi^2}{8n^2}. \qquad \qquad \qed
\end{eqnarray*}
\end{proof}



\subsection{Proof of Theorem \ref{ese_M}}
\label{sec:app:proof-4}
\begin{proof}
Monte Carlo method uses random directions to approximate the norm, which means
$$
u_i\sim \mbox{Unif}(S^1), i.i.d.
$$
We also know that
\begin{eqnarray}
&&\mathop{\mathbb{E}}\limits_{u_i, v\sim \mbox{Unif}(S^1)}\left\{\left|C_n\sum\limits_{i=1}^n\left|u_i^Tv\right|-1\right|^2\right\} \nonumber \\
&=&C_n^2\mathop{\mathbb{E}}\limits_{u_i, v\sim \mbox{Unif}(S^1)}\left\{\left(\sum\limits_{i=1}^n\left|u_i^Tv\right|\right)^2\right\}-2C_n\mathop{\mathbb{E}}\limits_{u_i,v\sim \mbox{Unif}(S^1)}\left\{\sum\limits_{i=1}^n\left|u_i^Tv\right|\right\}+1 \nonumber \\
&=&C_n^2\sum\limits_{i=1}^n\mathop{\mathbb{E}}\limits_{u_i, v\sim \mbox{Unif}(S^1)}\left(\left|u_i^Tv\right|^2\right)+2C_n^2\sum\limits_{1\leq i<j\leq N}\mathop{\mathbb{E}}\limits_{u_i, u_j, v\sim \mbox{Unif}(S^1)}\left(\left|u_i^Tv\right|\left|u_j^Tv\right|\right) \nonumber \\
&&-2C_n\sum\limits_{i=1}^n\mathop{\mathbb{E}}\limits_{u_i,v\sim \mbox{Unif}(S^1)}\left\{\left|u_i^Tv\right|\right\}+1, \label{eq_M}
\end{eqnarray}
where
$C_n$ satisfies
$$
C_n\cdot\int\limits_{u_i\in S^1}\sum\limits_{i=1}^n\left|u_i^Tv\right|du_i=1,
$$
which implies
$$
C_n = \frac{\pi}{2n}.
$$

We can find out the expected squared error if we can get the values of
\begin{eqnarray*}
&&\mathop{\mathbb{E}}\limits_{u_i, v\sim \mbox{Unif}(S^1)}\left(\left|u_i^Tv\right|^2\right),\\
&&\mathop{\mathbb{E}}\limits_{u_i, u_j, v\sim \mbox{Unif}(S^1)}\left(\left|u_i^Tv\right|\left|u_j^Tv\right|\right), \\
&&\mathop{\mathbb{E}}\limits_{u_i, v\sim \mbox{Unif}(S^1)}\left\{\left|u_i^Tv\right|\right\}, ~~~
\mbox{for all }i,j = 1, \cdots, n.
\end{eqnarray*}

Let $u_i=(\cos\phi, \sin\phi)^\prime, v=(\cos\theta,\sin\theta)^\prime,$ where $\phi\sim \mbox{Unif}(0,2\pi), \theta\sim \mbox{Unif}(0,2\pi).$
Then the above three can be computed as follows:
\begin{eqnarray*}
&&\mathop{\mathbb{E}}\limits_{u_i, v\sim \mbox{Unif}(S^1)}\left(\left|u_i^Tv\right|^2\right) \\
&=&\mathop{\mathbb{E}}\limits_{\phi,\theta\sim \mbox{Unif}(0,2\pi)}\cos^2(\phi-\theta)
=\mathop{\mathbb{E}}\limits_{\phi\sim \mbox{Unif}(0,2\pi)}\left[\mathop{\mathbb{E}}\limits_{\theta\sim \mbox{Unif}(0,2\pi)}\left[\cos^2(\phi-\theta)\vert\phi\right]\right] \\
&=&\frac{1}{2},
\end{eqnarray*}
and
\begin{eqnarray*}
&&\mathop{\mathbb{E}}\limits_{u_i, u_j, v\sim \mbox{Unif}(S^1)}\left(\left|u_i^Tv\right|\left|u_j^Tv\right|\right) \\
&=&\mathop{\mathbb{E}}\limits_{\phi_i,\phi_j,\theta\sim \mbox{Unif}(0,2\pi)}\left\{\left|\cos(\theta-\phi_i)\right|\left|\cos(\theta-\phi_j)\right|\right\} \\
&=&\mathop{\mathbb{E}}\limits_{\phi_j,\theta\sim \mbox{Unif}(0,2\pi)}\left\{\mathop{\mathbb{E}}\limits_{\phi_i\sim \mbox{Unif}(0,2\pi)}\left[\left|\cos(\theta-\phi_i)\right|[\left|\cos(\theta-\phi_j)\right|\vert\phi_j,\theta\right]\right\} \\
&=&\mathop{\mathbb{E}}\limits_{\phi_j,\theta\sim \mbox{Unif}(0,2\pi)}\left\{\left|\cos(\theta-\phi_j)\right|\cdot\frac{2}{\pi}\right\} \\
&=&\mathop{\mathbb{E}}\limits_{\theta\sim \mbox{Unif}(0,2\pi)}\left\{\mathop{\mathbb{E}}\limits_{\phi_j\sim \mbox{Unif}(0,2\pi)}\left[|\cos(\theta-\phi_j)|\cdot\frac{2}{\pi}\vert\theta\right]\right\}
=\mathop{\mathbb{E}}\limits_{\theta\sim \mbox{Unif}(0,2\pi)}\left[\frac{2}{\pi}\cdot\frac{2}{\pi}\right] \\
&=&\frac{4}{\pi^2},
\end{eqnarray*}
and
\begin{eqnarray*}
&&\mathop{\mathbb{E}}\limits_{u_i, v\sim \mbox{Unif}(S^1)}\left\{\left|u_i^Tv\right|\right\} \\
&=&\mathop{\mathbb{E}}\limits_{\phi,\theta\sim \mbox{Unif}(0,2\pi)}\left|\cos(\phi-\theta)\right|=\mathop{\mathbb{E}}\limits_{\phi\sim \mbox{Unif}(0,2\pi)}\left[\mathop{\mathbb{E}}\limits_{\theta\sim \mbox{Unif}(0,2\pi)}\left[\left|\cos(\phi-\theta)\right|\vert\phi\right]\right] \\
&=& \mathop{\mathbb{E}}\limits_{\phi\sim \mbox{Unif}(0,2\pi)}\frac{2}{\pi}
=\frac{2}{\pi}.
\end{eqnarray*}

Therefore by plugging the above results into \eqref{eq_M}, we eventually get
\begin{eqnarray*}
&&\mathop{\mathbb{E}}\limits_{u_i, v\sim \mbox{Unif}(S^1)}\left\{\left|C_n\sum\limits_{i=1}^n\left|u_i^Tv\right|-1\right|^2\right\} \\
&=&\frac{\pi^2}{4N^2}\left(N\cdot\frac{1}{2}+2\frac{N(N-1)}{2}\frac{4}{\pi^2}\right)-2\frac{\pi}{2n}\cdot N\cdot\frac{2}{\pi}+1
=\frac{\pi^2-8}{8N}. \qquad \qquad \qed
\end{eqnarray*}
\end{proof}

\subsection{Proof of Lemma \ref{V_max}}
\label{sec:app:proof-5}
\begin{proof}
Recall that we have
\begin{equation}\label{maxmax}
V_{\max} = \max\limits_{v:||v||_2=1}\sum\limits_{i=1}^n\left|u_i^Tv\right|
= \max\limits_{v:||v||_2=1}\max\limits_{s_i\in\{1, -1\}}\left(\sum\limits_{i=1}^ns_iu_i^T\right)v,
\end{equation}
where the second equality is based on a standard trick in optimization \cite[Chapter 9.2(ii)]{bradley1977applied}.

The following is an application of the Cauchy-Schwartz inequality:
$$
\left(\sum\limits_{i=1}^ns_iu_i^T\right)v\leq \sqrt{\left\|\sum\limits_{i=1}^n s_iu_i\right\|_2^2||v||_2^2}
= \left\|\sum\limits_{i=1}^n s_iu_i\right\|,
$$
where the equality is due to the condition $\left\|v\right\|=1$.

In the first part, the equality holds if and only if $|v_j| = c\left|\left(\sum\limits_{i=1}^ns_iu_i\right)_j\right|, j = 1,...,p$.

Apparently, we must have $c = \left\|\sum\limits_{i=1}^ns_iu_i\right\|^{-1}$ (because of $\left\|v\right\|=1$).

So we can have
\begin{equation}\label{v}
v= \frac{\sum\limits_{i=1}^ns_iu_i}{\left\|\sum\limits_{i=1}^ns_iu_i\right\|}.
\end{equation}
Combining \eqref{v} and \eqref{maxmax}, we have \eqref{Vmax}.
\qed
\end{proof}

\subsection{Proof of Lemma \ref{prop}}
\label{sec:app:proof-6}
\begin{proof}
We start with a special case: the linear subspace is $\mathbb{R}^p$(the entire space). Obviously the $n$ hyperplanes
$$\left\{y: u_i^Ty=0\right\}\mbox{, for } i = 1,2,...,n$$
divide the sphere $S^{p-1}$ into at most $2^n$ sectors. Within each sector, function$f(v)$ is strictly linear, therefore the minima cannot be an interior point.
Recall a boundary point $v$ must have $u_j^Tv=0$ for at least one $j, 1\leq j\leq n$.

Now we consider a linear subspace with dimension less than $p$, say, $k$.
Let $b_1,...,b_k$ be the orthonormal basis of such a linear subspace, we have $\forall x\in \Omega$,
$$
x = \sum\limits_{j=1}^kc_jb_j,
$$
and
$$
\sum\limits_{j=1}^kc_j^2=1, (\mbox{Because we have }\|x\|=1).
$$
Therefore, we have
$$
f(v) = \sum\limits_{i=1}^n\left|u_i^Tv\right|
= \sum\limits_{i=1}^n\left|u_i^T\sum\limits_{j=1}^kc_jb_j\right|
= \sum\limits_{i=1}^n\left|\sum\limits_{j=1}^kc_j(u_i^Tb_j)\right|
=  \sum\limits_{i=1}^n\left|h_i^Tc\right|,
$$
where $c = (c_1,...,c_k)^T$ and $h_i^T = \left(u_i^Tb_1,...,u_i^Tb_k\right), i=1, ..., n.$
Note that in the early part of this proof, the $u_i$ can be arbitrary.

The above derivation indicates that the latter case can be converted into the former case, as $c\in\mathbb{R}^k$ is from the entire space.
So we can get
$$
h_i^Tc = 0\mbox{ for at least one }i, 1 \leq i \leq n.
$$
As $h_i^Tc = u_i^T\left(\sum\limits_{j=1}^kb_jc_j\right)$, the above is equivalent to
$$
u_i^T\left(\sum\limits_{j=1}^kb_jc_j\right) =0\mbox{ for at least one }i, 1 \leq i \leq n.
$$
Quantity $\sum\limits_{j=1}^kb_jc_j$ can also be denoted as $v$, because any vector on the space is a linear combination of the orthonormal basis $b_1, ..., b_k.$

From all the above, we proved the lemma. \qed
\end{proof}

\subsection{Proof of Lemma \ref{Omega_0}}
\label{sec:app:proof-7}
\begin{proof}
For notational simplicity, let us donate $\Omega = \Omega(v_{\min})$. We can easily verify the following
$$
\mbox{rank}(\Omega)\leq p-1.
$$
Otherwise (i.e., $\mbox{rank}(\Omega) = p$), by the definition of $\Omega$, we will have $v_{\min} = 0.$
Now we show that
$$
\mbox{rank}(\Omega)\geq p-1.
$$
We use contradiction.
Let us assume that $\mbox{rank}(\Omega)<p-1.$
Define the following complementary set
$$
\Omega^{\perp} = \left\{x:\|x\| = 1, x\perp \Omega\right\},
$$
where $x\perp\Omega$ stands for that $x$ is perpendicular to the linear space that is spanned by all the $u_j$\rq s in $\Omega$.
Because $v_{\min}$ is a minimizer, we have that
$$
f(v_{\min}) = \min\limits_{v\in\Omega^\perp}f(v)
=  \min\limits_{v\in\Omega^\perp}\sum\limits_{i=1}^n\left|u_i^Tv\right|
= \min\limits_{v\in\Omega^\perp}\sum\limits_{u_i\not\in\Omega}\left|u_i^Tv\right|
$$
Note that if $\mbox{rank}(\Omega)<p-1$, we have $\mbox{dim}(\Omega^\perp)\geq 2$.

By Lemma \ref{prop}, we can declare that there exists $u_j\not\in\Omega$, $u_j^Tv_{\min}=0$.
However, this contradicts to the definition of $\Omega$, which is supposed to be the maximal subset. \qed
\end{proof}

\subsection{Proof of Theorem \ref{dim=n}}
\label{sec:app:proof-8}
\begin{proof}
When $n=p$, we have
$$
f(v) = |u_1^Tv| + |u_2^Tv| + ... + |u_p^Tv|, \mbox{ for } u_1,...u_p, v\in \mathbb{S}^{p-1}.
$$
According to the Lemma \ref{Omega_0}, we have
$$
\mbox{rank}\left(\Omega(v_{\min})\right) = p-1,
$$
where $\Omega(v_{\min}) = \left\{u_j:u_j^Tv_{\min}=0\right\}$, and $v_{\min}$ is the minimizer of $f(v).$
So the minimizer of $f(v)$ must satisfy that it is orthogonal to $p-1$ linearly independent $u_j$\rq s.

Assume every $p-1$ $u_j$\rq s are linearly independent.
Then the minimizer is among the vectors that are orthogonal to any $p-1$ $u_j$\rq s.
We know there are ${p \choose p-1}=p$ different combinations of $u_j$\rq s, and each combination is correspond to $2$ unit vectors orthogonal to one of the $p-1$ $u_j$\rq s.
(These $2$ unit vectors are the two directions that are orthogonal to a $p-1$ spaces in $\mathbb{R}^p.$)
Thus there are totally $2p$ unit vectors that might be the minimizer of $f(v).$

Suppose $p$ of the $2p$ unit vectors are those whose first nonzero entry is positive.
Denote them as $v^{-(1)}, v^{-(2)}, ..., v^{-(p)}$.
Then the other $p$ unit vectors would be $-v^{-(1)}, -v^{-(2)}, ..., -v^{-(p)}.$
Suppose that for any $i\in\{1,2,...,p\},$ $v^{-(1)}, v^{-(2)}, ..., v^{-(p)}$ satisfy
$$
\left(v^{-(i)}\right)^Tu_j =0, \forall j\not=i, j\in\{1,2,...,p\}.
$$
Thus the minimum value of $f(v)$ can be upper bounded by the average of the function values of the $p$ unit vectors:
\begin{equation}\label{fvmin}
\min_vf(v) \leq \frac{1}{p}\sum\limits_{i = 1}^pf(v^{-(i)}).
\end{equation}
We can also bound the maximum value of $f(v)$ by some value:
\begin{equation}\label{fvmax1}
\max_vf(v) \geq \max_{s_i = \pm1}f\left(\frac{\sum\limits_{i = 1}^ps_iv^{-(i)}}{\|\sum\limits_{i = 1}^ps_iv^{-(i)}\|}\right).
\end{equation}
Because we have
$$
f\left(\frac{\sum\limits_{i = 1}^ps_iv^{-(i)}}{\|\sum\limits_{i = 1}^ps_iv^{-(i)}\|}\right)
= \frac{f\left(\sum\limits_{i = 1}^ps_iv^{-(i)}\right)}{\|\sum\limits_{i = 1}^ps_iv^{-(i)}\|},
$$
and
\begin{eqnarray*}
f\left(\sum\limits_{i = 1}^ps_iv^{-(i)}\right)
&=& \sum\limits_{j=1}^p\left|u_j^T\left(\sum\limits_{i = 1}^ps_iv^{-(i)}\right)\right|
= \sum\limits_{j=1}^p\left|\sum\limits_{i = 1}^ps_iu_j^Tv^{-(i)}\right|
= \sum\limits_{j=1}^p\left|s_ju_j^Tv^{-(j)}\right| \\
&=& \sum\limits_{j=1}^p\left|u_j^Tv^{-(j)}\right| = \sum\limits_{j=1}^pf\left(v^{-(j)}\right),
\end{eqnarray*}
we can get
$$
f\left(\frac{\sum\limits_{i = 1}^ps_iv^{-(i)}}{\|\sum\limits_{i = 1}^ps_iv^{-(i)}\|}\right)
= \frac{\sum\limits_{j=1}^pf\left(v^{-(j)}\right)}{\|\sum\limits_{i = 1}^ps_iv^{-(i)}\|}.
$$
So \eqref{fvmax1} becomes
\begin{equation}\label{fvmax2}
\max_vf(v) \geq \max_{s_i = \pm1}\frac{\sum\limits_{i = 1}^pf(v^{-(i)})}{\left\|\sum\limits_{i = 1}^ps_iv^{-(i)}\right\|}
= \frac{\sum\limits_{i = 1}^pf(v^{-(i)})}{\min\limits_{s_i = \pm 1}\left\|\sum\limits_{i = 1}^ps_iv^{-(i)}\right\|}.
\end{equation}
Based on \eqref{fvmin} and \eqref{fvmax2},  we can get
$$
\frac{\min\limits_vf(v)}{\max\limits_vf(v)} \leq
\frac{\frac{1}{p}\sum\limits_{i = 1}^pf(v^{-(i)})}{\frac{\sum\limits_{i = 1}^pf(v^{-(i)})}{\min\limits_{s_i = \pm 1}\left\|\sum\limits_{i = 1}^ps_iv^{-(i)}\right\|}}
= \frac{1}{p}\min\limits_{s_i = \pm 1}\left\|\sum\limits_{i = 1}^ps_iv^{-(i)}\right\|.
$$
So we have
\begin{equation}\label{bound}
\max\limits_{u_1,...u_p}\frac{\min\limits_vf(v)}{\max\limits_vf(v)} \leq
\frac{1}{p}\max_{u_1,...u_p}\min\limits_{s_i = \pm 1}\left\|\sum\limits_{i = 1}^ps_iv^{-(i)}\right\|.
\end{equation}
Since solving the problem
$$
\max_{u_1,...u_p}\min\limits_{s_i = \pm 1}\left\|\sum\limits_{i = 1}^ps_iv^{-(i)}\right\|
$$
is equivalent to solving $$\max_{u_1,...u_p}\min\limits_{s_i = \pm 1}\left\|\sum\limits_{i = 1}^ps_iv^{-(i)}\right\|^2,
$$
 we will try to solve the latter one in the following.
 We have
$$
\max_{u_1,...u_p}\min\limits_{s_i = \pm 1}\left\|\sum\limits_{i = 1}^ps_iv^{-(i)}\right\|^2 = \max_{u_1,...u_p}\min\limits_{s_i = \pm 1}s^T\Sigma s,
$$
where we have $\Sigma \in\mathbb{R}^{p\times p}$ and
\begin{eqnarray*}
\Sigma &=& \left(\begin{array}{ccccc}
1& \left(v^{-(1)}\right)^Tv^{-(2)} & \left(v^{-(1)}\right)^Tv^{-(3)} &  \cdots & \left(v^{-(1)}\right)^T v^{-(p)} \\
\left(v^{-(2)}\right)^T v^{-(1)} & 1 & \left(v^{-(2)}\right)^Tv^{-(3)}  & \cdots &\left(v^{-(2)}\right)^T v^{-(p)} \\
\cdots&\cdots&\cdots&\cdots&\cdots \\
\left(v^{-(p)}\right)^T v^{-(1)} & \left(v^{-(p)}\right)^Tv^{-(2)}  & \left(v^{-(p)}\right)^Tv^{-(3)} & \cdots &1 \\
\end{array}
\right).
\end{eqnarray*}
We claim that $\min\limits_{s_i = \pm 1}s^T\Sigma s$ is upper bounded by $p$, and $\min\limits_{s_i = \pm 1}s^T\Sigma s=p$ when
$$
\left(v^{-(i)}\right)^Tv^{-(j)}=0, \forall i\not=j.
$$

We can see that if there are some $i,j$ $(i\not=j)$, such that $\left(v^{-(i)}\right)^Tv^{-(j)}\not=0,$ then there exists some $s$, such that $s^T\Sigma s\leq p$.
Suppose there does not exist such $s$, which means for any $s$, the following holds,
\begin{equation}\label{eq:sSs}
s^T\Sigma s>p.
\end{equation}
Since we have
\begin{eqnarray*}
\sum\limits_{s_i = \pm 1} s^T\Sigma s
&=&\sum\limits_{s\in\{s:s_k=\pm1\}}\sum\limits_{i,j}s_is_j\Sigma_{ij}
= \sum\limits_{s\in\{s:s_k=\pm1\}}\left(p + \sum\limits_{i\not=j}s_is_j\Sigma_{ij}\right) \\
&=& 2^pp +  \sum\limits_{s\in\{s:s_k=\pm1\}}\sum\limits_{i\not=j}s_is_j\Sigma_{ij} = 2^pp,
\end{eqnarray*}
this will lead to $\sum\limits_{s_i = \pm 1} s^T\Sigma s>2^pp,$ which is a contradiction of \eqref{eq:sSs}.
So we proved that our claim is true, which says
$$
\min\limits_{s_i = \pm 1}s^T\Sigma s\leq p,
$$
and when $\left(v^{-(i)}\right)^Tv^{-(j)}=0, \forall i\not=j,$ which means $\Sigma = I_p$, we have
$\min\limits_{s_i = \pm 1}s^T\Sigma s =p.$

We know that $v^{-(i)}$\rq s only depends on $u_i$\rq s, and when
$u_i^Tu_j = 0, \forall i\not=j,$ we have
$\left(v^{-(i)}\right)^Tv^{-(j)}=0, \forall i\not=j.$
So when the following holds,
$$
u_i^Tu_j = 0, \forall i\not=j,
$$
$\min\limits_{s_i = \pm 1}s^T\Sigma s$ achieves the maximum value, which is $p$.
Therefore we get
$$
\max_{u_1,...u_p}\min\limits_{s_i = \pm 1}\left\|\sum\limits_{i = 1}^ps_iv^{-(i)}\right\|^2
= \max_{u_1,...u_p}\min\limits_{s_i = \pm 1}s^T\Sigma s =p,
$$
which leads to
\begin{equation}\label{eq:maxminnorms}
\max_{u_1,...u_p}\min\limits_{s_i = \pm 1}\left\|\sum\limits_{i = 1}^ps_iv^{-(i)}\right\| = \sqrt{p}.
\end{equation}
Based on \eqref{bound} and \eqref{eq:maxminnorms}, we have
\begin{equation}\label{upper}
\max\limits_{u_1,...u_p}\frac{\min\limits_vf(v)}{\max\limits_vf(v)} \leq \frac{\sqrt{p}}{p}.
\end{equation}

Next if we can prove that  when $u_i^Tu_j = 0, \forall i\not=j, $
the following holds,
$\frac{\min\limits_vf(v)}{\max\limits_vf(v)} = \frac{\sqrt{p}}{p}$;
combined with \eqref{upper}, we can arrive at the conclusion and finish the proof of the Lemma.

Let us assume
$$
u_i^Tu_j = 0, \forall i\not=j.
$$
Without loss of generality, we can assume $u_i = e_i, \forall i\not=j,$ where $e_i$\rq s are the basic vectors of $\mathbb{R}^p.$
Then the following holds,
$$
f(v) = \sum\limits_{i=1}^p|v_i|, v\in \mathbb{S}^{p-1}.
$$
We can easily verify the following,
$\min\limits_vf(v) = 1,$ and $\max\limits_vf(v) = \sqrt{p}.$
So when $u_i^Tu_j = 0, \forall i\not=j$, we have
$$
\frac{\min\limits_vf(v)}{\max\limits_vf(v)} = \frac{\sqrt{p}}{p}.
$$
Combined what we get from \eqref{upper}, that is, $\frac{\sqrt{p}}{p}$ is the upper bound of $\max\limits_{u_1,...u_p}\frac{\min\limits_vf(v)}{\max\limits_vf(v)},$ we finished the proof. \qed
\end{proof}


\subsection{Proof of Lemma \ref{u_1}}
\label{sec:app:proof-9}
\begin{proof}
As we have
$$
\min_{x:\|x\|=1,\left<x,v\right>=\theta}\left\|x+B\right\|^2
=  \min_{x:\|x\|=1,\left<x,v\right>=\theta}1 + \|B\|^2 + 2\left<x,B\right>,
$$
the problem \eqref{u1} is equivalent to
\begin{equation}\label{u1_2}
\min_{x:\|x\|=1,\left<x,v\right>=\theta}\left<x,B\right>.
\end{equation}
Suppose $x^*$ is the solution to the above problem \eqref{u1}. Then $x^*$ is the farthest point to $B$ on the circle that satisfies the constraints $\|x\|=1,\left<x,v\right>=\theta$. The three points $x^*, v, \mbox{ and }B$ must be on a same plane. Therefore, we can assume
\begin{equation}\label{eq:xstar}
x^* = av+bB.
\end{equation}
Bringing \eqref{eq:xstar} into \eqref{u1_2}, we have
\begin{equation*}\label{u1_3}
\min_{x:\|x\|=1,\left<x,v\right>=\theta}\left<x,B\right>=\min_{a,b:\|av+bB\|=1,\left<av+bB,v\right>=\theta}\left<av+bB,B\right>,
\end{equation*}
which is equivalent to
\begin{eqnarray}
\min_{a,b} & av^TB+bB^TB \label{opt}\\
\text{s.t.} &
 \begin{cases}\label{cond}
    a^2 + b^2\|B\|^2 + 2abv^TB&=1\\
    a+bv^TB& =\cos\theta.\\
 \end{cases}
\end{eqnarray}
Bringing the second equation in the constraints \eqref{cond}, that is,
\begin{equation}\label{a}
a = \cos\theta - bv^TB
\end{equation}
into \eqref{opt}, we have
\begin{eqnarray}
\label{opt2}
\min_{b} && v^TB\cos\theta+b\left(B^TB-(v^TB)^2\right) \nonumber \\
\text{s.t.}&&
 b^2\left(B^TB-(v^TB)^2\right) + \cos\theta^2=1.
\end{eqnarray}
Then the solution to \eqref{opt2} is
\begin{equation*}\label{b}
b = \pm\frac{\sin\theta}{\sqrt{B^TB-(v^TB)^2}}.
\end{equation*}
Since $B^TB-(v^TB)^2\geq 0$, the minimum is achieved when
\begin{equation}\label{b_final}
b = -\frac{|\sin\theta|}{\sqrt{B^TB-(v^TB)^2}}.
\end{equation}
Combining \eqref{b_final} with \eqref{a}, we can get the solution. \qed
\end{proof}

\subsection{Proof of Theorem \ref{gtheta}}
\label{sec:app:proof-10}
\begin{proof}
If $\theta \in [0,\pi)$, the square of the denominator of \eqref{x} becomes
$$
1 + 2v^TB\cos\theta + B^TB - 2\sin\theta\sqrt{B^TB-(v^TB)^2}
= 1 + B^TB + 2\sqrt{B^TB}\sin(\alpha - \theta),
$$
where
\begin{eqnarray*}
\sin\alpha &=& \frac{v^TB}{\sqrt{B^TB}}, \\
\cos\alpha &=& \frac{\sqrt{B^TB-(v^TB)^2}}{\sqrt{B^TB}}.
\end{eqnarray*}
Similarly, if $\theta \in [\pi, 2\pi)$, then the square of the denominator of \eqref{x} becomes
$$
1 + 2v^TB\cos\theta + B^TB + 2\sin\theta\sqrt{B^TB-(v^TB)^2}
= 1 + B^TB + 2\sqrt{B^TB}\sin(\alpha + \theta),
$$
where $\alpha$ is the same defined as above.

Hence, for $\theta \in [0,\pi)$, we have
$$
f(\theta) = \frac{|\cos\theta| + A}{\sqrt{1 + B^TB + 2\sqrt{B^TB}\sin(\alpha - \theta)}};
$$
for  $\theta \in [\pi,2\pi)$, we have
$$
f(\theta) = \frac{|\cos\theta| + A}{\sqrt{1 + B^TB + 2\sqrt{B^TB}\sin(\alpha + \theta)}},
$$
which is equivalent to
$$
f(\theta) = \frac{|\cos\theta| + A}{\sqrt{1 + B^TB + 2\sqrt{B^TB}\sin(\alpha + \theta)}},
$$
where $\theta \in [-\pi,0),$
which is also equivalent to
$$
f(\theta) = \frac{|\cos\theta| + A}{\sqrt{1 + B^TB + 2\sqrt{B^TB}\sin(\alpha - \theta)}},
$$
where $\theta \in [0,\pi).$

So the problem we want to solve is actually to maximize
$$
f(\theta) = \frac{|\cos\theta| + A}{\sqrt{1 + B^TB + 2\sqrt{B^TB}\sin(\alpha - \theta)}}
$$
on $\theta \in [0,\pi)$.

Under the first order condition, we have that if $\theta^*$ maximizes $f(\theta)$, then
$0 = f\rq(\theta^*).$

When $\theta \in [0,\frac{\pi}{2})$, the first order differentiable function of $f(\theta)$ can be written as
$$
f\rq(\theta) = \frac{-(1+B^TB)\sin\theta + \sqrt{B^TB}\left[\cos\alpha
+ A\cos(\alpha-\theta)
- \sin\theta\sin(\alpha - \theta)\right]}{\left(1 + B^TB + 2\sqrt{B^TB}\sin(\alpha - \theta)\right)^{3/2}};
$$
When $\theta \in [\frac{\pi}{2},\pi)$, the first order differentiable function of $f(\theta)$ can be written as
$$
f\rq(\theta) = \frac{(1+B^TB)\sin\theta
+ \sqrt{B^TB}\left[-\cos\alpha + A\cos(\alpha-\theta)
+ \sin\theta\sin(\alpha - \theta)\right]}{\left(1 + B^TB + 2\sqrt{B^TB}\sin(\alpha - \theta)\right)^{3/2}}.
$$

If we define function $g(\theta)$ as the following
\begin{equation*}
g(\theta) =
\left\{ \begin{array}{ll}
\sqrt{B^TB}\left[\cos\alpha
+ A\cos(\alpha-\theta) - \sin\theta\sin(\alpha - \theta)\right] & \\
-(1+B^TB)\sin\theta, & \mbox{ if } \theta\in[0,\frac{\pi}{2}),\\
\sqrt{B^TB} \left[-\cos\alpha + A\cos(\alpha-\theta) + \sin\theta\sin(\alpha - \theta)\right] & \\
+ (1+B^TB)\sin\theta  &\mbox{ if } \theta\in[\frac{\pi}{2},\pi),
\end{array}
\right.
\end{equation*}
Then our goal becomes to find the zeros of the function $g(\theta).$ \qed
\end{proof}

\bibliographystyle{plainnat}
\bibliography{author}

\begin{thebibliography}{19}
\providecommand{\natexlab}[1]{#1}
\providecommand{\url}[1]{\texttt{#1}}
\expandafter\ifx\csname urlstyle\endcsname\relax
  \providecommand{\doi}[1]{doi: #1}\else
  \providecommand{\doi}{doi: \begingroup \urlstyle{rm}\Url}\fi

\bibitem[Asmussen and Glynn(2007)]{asmussen2007stochastic}
S{\o}ren Asmussen and Peter~W Glynn.
\newblock \emph{Stochastic simulation: algorithms and analysis}, volume~57.
\newblock Springer Science \& Business Media, 2007.

\bibitem[Bradley et~al.(1977)Bradley, Hax, and Magnanti]{bradley1977applied}
Stephen Bradley, Arnoldo Hax, and Thomas Magnanti.
\newblock Applied mathematical programming.
\newblock 1977.

\bibitem[Brauchart et~al.(2014)Brauchart, Saff, Sloan, and
  Womersley]{brauchart2014qmc}
Johann Brauchart, E~Saff, I~Sloan, and R~Womersley.
\newblock Qmc designs: optimal order quasi monte carlo integration schemes on
  the sphere.
\newblock \emph{Mathematics of computation}, 83\penalty0 (290):\penalty0
  2821--2851, 2014.

\bibitem[Chaudhuri and Hu(2018)]{chaudhuri2018fast}
Arin Chaudhuri and Wenhao Hu.
\newblock A fast algorithm for computing distance correlation.
\newblock \emph{arXiv preprint arXiv:1810.11332}, 2018.

\bibitem[Hesse et~al.(2010)Hesse, Sloan, and Womersley]{hesse2010numerical}
Kerstin Hesse, Ian~H Sloan, and Robert~S Womersley.
\newblock Numerical integration on the sphere.
\newblock In \emph{Handbook of Geomathematics}, pages 1185--1219. Springer,
  2010.

\bibitem[Hoeffding(1992)]{hoeffding1992class}
Wassily Hoeffding.
\newblock A class of statistics with asymptotically normal distribution.
\newblock In \emph{Breakthroughs in Statistics}, pages 308--334. Springer,
  1992.

\bibitem[Huang and Huo(2017)]{huang2017efficient}
Cheng Huang and Xiaoming Huo.
\newblock An efficient and distribution-free two-sample test based on energy
  statistics and random projections.
\newblock \emph{arXiv preprint arXiv:1707.04602}, 2017.

\bibitem[Huo and Sz{\'e}kely(2016)]{huo2016fast}
Xiaoming Huo and G{\'a}bor~J Sz{\'e}kely.
\newblock Fast computing for distance covariance.
\newblock \emph{Technometrics}, 58\penalty0 (4):\penalty0 435--447, 2016.

\bibitem[Korolyuk and Borovskich(2013)]{korolyuk2013theory}
Vladimir~S Korolyuk and Yu~V Borovskich.
\newblock \emph{Theory of U-statistics}, volume 273.
\newblock Springer Science \& Business Media, 2013.

\bibitem[Lyons et~al.(2013)]{lyons2013distance}
Russell Lyons et~al.
\newblock Distance covariance in metric spaces.
\newblock \emph{The Annals of Probability}, 41\penalty0 (5):\penalty0
  3284--3305, 2013.

\bibitem[Mises(1947)]{mises1947asymptotic}
R~v Mises.
\newblock On the asymptotic distribution of differentiable statistical
  functions.
\newblock \emph{The annals of mathematical statistics}, 18\penalty0
  (3):\penalty0 309--348, 1947.

\bibitem[Morokoff and Caflisch(1995)]{morokoff1995quasi}
William~J Morokoff and Russel~E Caflisch.
\newblock Quasi-monte carlo integration.
\newblock \emph{Journal of computational physics}, 122\penalty0 (2):\penalty0
  218--230, 1995.

\bibitem[Nesterov(2012)]{nesterov2012efficiency}
Yu~Nesterov.
\newblock Efficiency of coordinate descent methods on huge-scale optimization
  problems.
\newblock \emph{SIAM Journal on Optimization}, 22\penalty0 (2):\penalty0
  341--362, 2012.

\bibitem[Niederreiter(1992)]{niederreiter1992random}
Harald Niederreiter.
\newblock \emph{Random number generation and quasi-Monte Carlo methods},
  volume~63.
\newblock Siam, 1992.

\bibitem[Sloan and Womersley(2004)]{sloan2004extremal}
Ian~H Sloan and Robert~S Womersley.
\newblock Extremal systems of points and numerical integration on the sphere.
\newblock \emph{Advances in Computational Mathematics}, 21\penalty0
  (1-2):\penalty0 107--125, 2004.

\bibitem[Sz{\'e}kely and Rizzo(2004)]{szekely2004testing}
G{\'a}bor~J Sz{\'e}kely and Maria~L Rizzo.
\newblock Testing for equal distributions in high dimension.
\newblock \emph{InterStat}, 5:\penalty0 1--6, 2004.

\bibitem[Sz{\'e}kely and Rizzo(2009)]{szekely2009brownian}
G{\'a}bor~J Sz{\'e}kely and Maria~L Rizzo.
\newblock Brownian distance covariance.
\newblock \emph{The annals of applied statistics}, pages 1236--1265, 2009.

\bibitem[Sz{\'e}kely et~al.(2007)Sz{\'e}kely, Rizzo, Bakirov,
  et~al.]{szekely2007measuring}
G{\'a}bor~J Sz{\'e}kely, Maria~L Rizzo, Nail~K Bakirov, et~al.
\newblock Measuring and testing dependence by correlation of distances.
\newblock \emph{The annals of statistics}, 35\penalty0 (6):\penalty0
  2769--2794, 2007.

\bibitem[Wright(2015)]{wright2015coordinate}
Stephen~J Wright.
\newblock Coordinate descent algorithms.
\newblock \emph{Mathematical Programming}, 151\penalty0 (1):\penalty0 3--34,
  2015.

\end{thebibliography}

\end{document}